\documentclass[conference]{IEEEtran}
\IEEEoverridecommandlockouts
\usepackage{balance}
\usepackage{cite}
\usepackage{amsmath,amssymb,amsfonts}
\usepackage{algorithmic}
\usepackage{graphicx}
\usepackage{textcomp}
\usepackage{xcolor}

\usepackage{hyperref}

\usepackage{amsthm}
\usepackage{booktabs}
\usepackage{multirow}
\usepackage[linesnumbered,ruled,vlined]{algorithm2e} 
\usepackage{algorithm2e,setspace}
\usepackage{subfigure}
\usepackage{color}
\usepackage{float}
\usepackage{makecell}
\usepackage{capt-of}
\usepackage{enumitem}
\usepackage{color}
\usepackage{colortbl}
\definecolor{gray}{rgb}{0.5,0.5,0.5}
\hypersetup{pdfauthor={Name}}

\theoremstyle{definition}
\newtheorem{definition}{Definition}[section]
\newcommand{\rev}[1]{\textcolor[rgb]{1,0,0}{#1}}
\theoremstyle{plain}
\newtheorem{lemma}{Lemma}[section]

\setitemize[1]{itemsep=0pt,partopsep=0pt,parsep=\parskip,topsep=5pt}

\theoremstyle{plain}
\newtheorem{theorem}{Theorem}[section]

\newtheorem{example}{\textbf{Example}}
\newcommand{\eat}[1]{#1}
\newcommand{\eatYou}[1]{}

\def\BibTeX{{\rm B\kern-.05em{\sc i\kern-.025em b}\kern-.08em
    T\kern-.1667em\lower.7ex\hbox{E}\kern-.125emX}}
\begin{document}

\title{Towards Real-Time Counting Shortest Cycles on Dynamic Graphs: A Hub Labeling Approach\\
\thanks{You Peng is the joint first and corresponding author.}}
\makeatletter
\newcommand{\linebreakand}{%
  \end{@IEEEauthorhalign}
  \hfill\mbox{}\par
  \mbox{}\hfill\begin{@IEEEauthorhalign}
}
\makeatother
\author{{Qingshuai Feng$^{\dagger}$, You Peng$^{\dagger}$$^*$, Wenjie Zhang$^{\dagger}$, Ying Zhang$^{\S}$, Xuemin Lin$^{\dagger}$}

\vspace{1.6mm}\\
\fontsize{10}{10}
\selectfont\itshape
$^\dagger$The University of New South Wales; $^\S$University of Technology Sydney\\
\fontsize{9}{9} \selectfont\ttfamily\upshape
q.feng.1@student.unsw.edu.au;
unswpy@gmail.com;\\
\{zhangw,lxue\}@cse.unsw.edu.au;
ying.zhang@uts.edu.au\\}


\maketitle

\begin{abstract}

With the ever-increasing prevalence of graph data in a wide spectrum of applications, it becomes essential to analyze structural trends in dynamic graphs on a continual basis. The shortest cycle is a fundamental pattern in graph analytics. In this paper, we investigate the problem of shortest cycle counting for a given vertex in dynamic graphs in light of its applicability to problems such as fraud detection. To address such queries efficiently, we propose a 2-hop labeling based algorithm called \underline{C}ounting \underline{S}hortest \underline{C}ycle (CSC for short). Additionally, techniques for dynamically updating the CSC index are explored. Comprehensive experiments are conducted to demonstrate the efficiency and effectiveness of our method. In particular, CSC enables query evaluation in a few hundreds of microseconds for graphs with millions of edges, and improves query efficiency by two orders of magnitude when compared to the baseline solutions. Also, the update algorithm could efficiently cope with edge insertions (deletions).
\end{abstract}


\section{Introduction}
\label{sct:introduction}
With the rapid development of information technology, a growing number of applications represent data as graphs~\cite{peng2018efficient, jin2021fast, peng2021answering,yuan2022efficient,peng2022finding,chen2022answering}. The graph structure naturally encodes complex relationships among entities in real-world networks such as social networks, e-commerce networks, and electronic payments networks. Sophisticated analytics over such graphs provides insights into the underlying dataset and interactions~\cite{peng2020answering, lai2021pefp, peng2021efficient}. Such networks often include millions of edges and vertices. 
Additionally, the structural changes to these networks are constant in nature, which makes them extremely dynamic. At scale, it is imperative and challenging to support real-time analytics on rapidly changing structural patterns in dynamic graphs~\cite{peng2021dlq, qiu18}.

Cycles are fundamental to graph analytics. A cycle is a path with at least 3 vertices and returns to its starting vertex. A cycle captures the circular structural pattern from the starting vertex to itself and is an informative indicator for graph analytics. In light of this, cycle detection has been extensively investigated in the literature for both static and dynamic graphs~\cite{qiu18, weinblatt72}. The shortest cycle refers to the cycle in a graph with the shortest length (i.e., the minimum number of edges) and the length is also called {\em girth} of the graph. The shortest cycle is often employed in graph structure analysis, for example, when solving the graph coloring problem ~\cite{thomassen83,gimbel97}. 

\begin{figure}[t]
  \centering
  \includegraphics[scale=0.7]{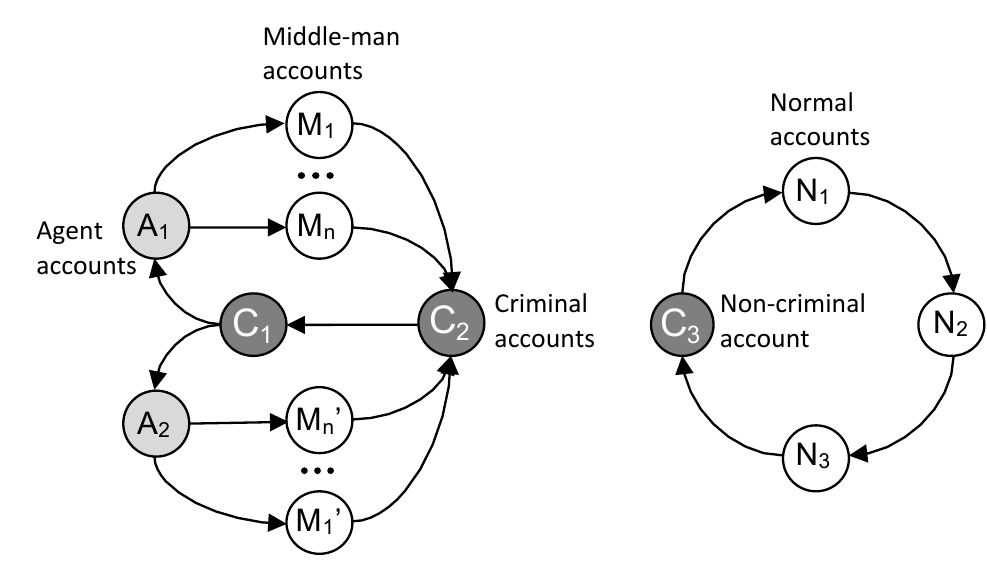}
  \vspace{-3mm}
  \caption{A Graph of a Money-Laundering Example.}
  \vspace{-3mm}
  \label{fig:application}
\end{figure}

The shortest cycles across a vertex $v$ denotes the cycles through $v$ with the shortest length \cite{weinblatt72,Yuster11}. For a particular vertex, the shortest cycles provide the quickest feedback routes originated at the vertex \cite{bonneau2017distribution}. The lengths of the shortest cycles have been studied to determine the frequencies of oscillations in neural circuits \cite{vladimirov2012shortest}. The distribution of shortest cycle lengths supports the structure analysis on chemical networks and biological networks \cite{gleiss2001relevant,klamt2009computing}, as well as feedback processes \cite{zanudo2017structure} and synchronization \cite{barrat2008dynamical} in various networks. 
While previous works mostly focus on computing the length of shortest cycles, in many real-world graphs, the diameter of the graphs tends to be small due to the small-world phenomenon. This leads to the fact that the shortest cycles of many vertices share the same length. For instance, Figure 1 depicts a transaction network involved in money laundering. The shortest cycle length through both vertex $C_1$ and vertex $C_3$ is 4. Nonetheless, there are significantly more shortest cycles via $C_1$. The number of shortest cycles via a vertex becomes a more informative metric for identifying suspicious users in transaction networks. Motivated by this, we investigate the problem of calculating shortest cycles for any given vertex in a graph.

\vspace{1mm}
{\bf Applications.}~Shortest cycle counting is required in a broad variety of applications. Two example applications are given below.


\vspace{1mm}
\textbf{Application 1: Fraud Detection.} \autoref{fig:application}~illustrates an example of money-laundering crime. Dark gray vertices $C_{1}$ and $C_{2}$ are both criminal accounts. Other vertices represent intermediary accounts (white nodes) or agents (light gray nodes) which assist the criminals. Directed edges depict forged transactions or other money laundering activities. The shortest cycles reflect criminals' most efficient ways for money laundering. As a result, a high number of such cycles travel through the criminal accounts. For example, a cycle, starting from and ending with $C_{1}$ via $A_{1}$ $M_{1}$ $C_{2}$, demonstrates such a route. The more such cycles that pass through a vertex, the more probable it is that the individual will engage in money-laundering activities. 
A pre-screening criterion for offenders might include a specified number of shortest cycles~\cite{qiu18,yue07}. Additionally, such data can also be utilized to detect fraudulent transactions on an e-commerce platform~\cite{qiu18,peng19}. 
\cite{vsubelj2011expert} develops an expert system for automobile insurance fraud detection. A collision graph is constructed by considering automobiles as vertices and edges between two vehicles as collisions. 
Fraudulent components share some structural characteristics, one of which is the occurrence of short cycles. \cite{bodaghi2018automobile} studies the similar problem. Their experiments demonstrate that detecting cycles is more effective in finding potential fraudulent groups than using other methods, such as leading eigenvector.  \cite{hajdu2020temporal} applies the cycle detection problem to a fraud detection system of financial institutions and demonstrates its effectiveness.

\vspace{1mm}
\textbf{Application 2: File Sharing Optimization.}
Consider another use case in peer-to-peer file-sharing networks (e.g., Gnutella)~\cite{schlosser03}. The vertices in the network represent hosts, while edges indicate the interactions, such as file request or transfer, between the hosts. A file-sharing cycle via a host signals the end of the file-sharing activity. For instance, host A sends a file request, which is propagated until it reaches host B. B is the location of the requested file and sends the file to host A. Generally, the most efficient routing is chosen with the minimum hop-count. The number of shortest cycles through host A along with the length might be used to determine the host's location. For instance, 
if host A has longer and more shortest cycles, a proxy server is probably required to decrease the transmission costs.
~\cite{ohta2004index} introduces an index-server optimization problem for P2P file sharing. They analyzed a flooding scheme and two index-server schemes 
for history and latest queries. In this problem, they need to set the index-server, while a machine with a high number of shortest cycles is preferred due to 1) these systems need to be failure-tolerant and 2) a needed file 
is easy to locate. The shortest cycles naturally satisfied these two requirements. 

In each of the aforementioned applications, the networks are extremely dynamic in nature.  For instance, in an e-commerce network, fraud accounts often initiate new transactions and transfer funds to others. New file requests and file transfers often occur in a file-sharing network. An update will be reflected in the graph as an edge insertion or deletion and will impact the query results for many vertices in the graph. In many scenarios, a static graph or a snapshot of a dynamic graph is insufficient, since continuous monitoring of shortest cycles numbers is needed. As a consequence, the dynamic issue is essential for shortest cycle counting. 

\vspace{1mm}
{\bf Challenges and Contributions.}
Numerous applications of the shortest cycle counting, particularly in online situations, demand real-time response to the query~\cite{qiu18,peng19}. Thus, it is critical to develop a real-time algorithm capable of counting the shortest cycles. A straightforward solution is to calculate the number of shortest cycles for each vertex in advance and record the values. Then, any query can be answered with $O(1)$ time complexity. Nevertheless, such a simple approach cannot handle dynamic graphs well since it requires to re-compute the shortest cycles for all vertices regarding graph updates. This is because of the lack of awareness for distance information among vertices within the graph. 
Hence, it is challenging to build an index that supports both real-time query answering and fast updates in the face of dynamic changes. 

\eatYou{In large-scale graphs, hub labeling is widely used to solve shortest path \cite{abraham11, abraham12, li17}, shortest distance \cite{jiang14,jin12,akiba13}, and reachability \cite{yano13} problems. It generates labels for each vertex in the graph in advance and evaluates queries using these labels. This method normally requires an offline index construction process and provides speedy query evaluation. Pruned landmark labeling \cite{akiba13} is the state-of-the-art technique for constructing hub labeling with full vertex ordering, in which a pre-defined vertex order determines whether a vertex will be accessed to update its label or can be securely pruned. Due to the fast evaluation, we propose a hub-pushing labeling technique, \underline{c}ounting \underline{s}hortest \underline{c}ycle (CSC for short), to build an index for shortest cycle counting. While hub labeling enables real-time query response, its dynamic maintenance is highly challenging. 
Existing hub labeling updating techniques place a premium on shortest path queries or reachability queries. Note that shortest cycle counting through a vertex v can be converted to shortest path counting between $v$ and its neighbours (See Section~\ref{sec:espc} for details). Nonetheless, maintaining counting queries is more challenging because even the shortest distance between two vertices remains constant, the counting information may still need to be updated. }

We address the aforementioned issues in this paper and offer the following contributions. 

\noindent
{\em (1) Hub Labeling algorithm (CSC) for Counting Shortest Cycles.} We devise a new and effective bipartite conversion-based technique for reshaping the initial graph and then constructing its hub labeling index CSC to enable real-time shortest cycle counting. Efficient pruning strategies are developed to expedite the index building process. An index merge mechanism is applied to aggregate related label entries and minimize the label size. Even if the bipartite conversion doubles the number of vertices in the reshaped graph, the new index remains a similar size compared with the baseline.

\noindent
{\em (2) Dynamic Maintenance of CSC.} We provide the first incremental (decremental) updating algorithm for an added (removed) edge in order to maintain our CSC index. Rather than rebuilding the whole index, only the portions impacted by the new edge are updated, substantially lowering the cost of maintenance.

\noindent
{\em (3) Comprehensive Experimental Evaluation.} We perform experiments on nine graphs to validate the effectiveness and efficiency of our algorithms. The experimental results demonstrate that CSC is comparable to the baseline in terms of index construction time and index size but is up to two orders of magnitude faster in terms of query processing. Also, the update algorithm could efficiently cope with edge insertions (deletions).


\vspace{1mm}
\noindent
{\bf Roadmap.} The rest of the paper is organized as follows. Section~\ref{sect:pre} introduces some preliminaries and Section~\ref{sec:baselines} introduces baselines. Our 2-hop based method is proposed in Section~\ref{sect:ours}. Section~\ref{sect:insert} investigates the index maintenance for edge insertion, followed by empirical studies in Section~\ref{sect:exp}. Section~\ref{sec:related} surveys important related work. Section~\ref{sect:conclusion} concludes the paper.

\section{Preliminaries}
\label{sect:pre}

\begin{table}[htb]
\centering
\caption{Notations}
\vspace{-2mm}
  \begin{tabular}{|l|l|}
    \hline
    \cellcolor{gray!25}\textbf{Notations} & \cellcolor{gray!25}\textbf{Description}\\
    \hline
     $G = (V, E)$ & \makecell[l] {directed graph with vertex set $V$ and edge \\ set $E$}\\
    \hline
    SP($s, t$) & all the shortest paths from $s$ to $t$\\
    \hline
    SPV($s, t$) & all the corresponding vertices in SP($s, t$)\\
    \hline
    SPCnt($s, t$) & the number of shortest paths from $s$ to $t$\\
    \hline
    SCCnt($v$) & the number of shortest cycles through $v$ \\
    \hline
    $G_{0}$ & the original graph\\
    \hline
    $G_{+} / G_{-}$ & the graph after an edge insertion/deletion\\
  \hline
\end{tabular}
\label{tab:notation}
\vspace{-4mm}
\end{table}

\subsection{Preliminaries}
\label{subsect:pre}
$G = (V, E)$ denotes a directed graph  where $V$ is the set of its vertices and $E \subseteq V \times V$ is the set of edges. $e(v,u)$ represents a directed edge connecting the vertex $v$ to vertex $u$. Let $n$ and $m$ represent the number of vertices and edges, respectively. The degree of vertex $v$ is defined as degree($v$), which is the sum of its in- and out-degrees. We use $nbr_{out}(v)$ to refer to its out neighbors or successors and $nbr_{in}(v)$ to refer to its in neighbors or ancestors. A path $p$ from vertex $s$ to vertex $t$ is a sequence of vertices in the form of $<s = v_{0}, v_{2},..., v_{k} = t>$ where for every  $i \in [1,k]$, $e(v_{i-1},v_{i})\in E$. $p(u,v)$ denotes a path from vertex $u$ to $v$. The length of a path $p$, indicated by $len(p)$, is the total number of its edges. A path is considered {\em simple} if it has no repeating vertices or edges. 
A path from vertex $s$ to vertex $t$ is said to be shortest if its length is equal to or less than the length of any other path from $s$ to $t$. The length of the shortest path from $s$ to $t$ is represented by $sd(s, t)$, and one such path is denoted by $sp(s, t)$. SP($s, t$) denotes the collection of all such shortest paths. SPV($s, t$) is the set of all vertices that correspond to SP($s, t$). Let SPCnt$(s, t)$ represent the number of shortest paths connecting $s$ and $t$, i.e., $|SP(s,t)|$. 

Similarly, we present the definitions for cycles. A {\em cycle} is a path $p$ with $v_0$ = $v_k$ and $len(p) \ge$ 3. A cycle is said to be {\em simple} if it has no repeats of vertices or edges other than the starting and ending vertex. The length of a cycle is equal to its number of edges. A cycle through a vertex $v$ is said to be shortest if its length is not larger than any other cycle through $v$. The number of shortest cycles through $v$ is denoted by SCCnt($v$).  

To tackle the dynamic graph, we examine only edge updates in the paper, since the insertion or deletion of a vertex can be represented by a series of edge insertions or deletions. We refer to the original graph as $G_{0}$, and use $G_{+}$ ($G_{-}$) to denote the graph after an edge insertion (deletion). \autoref{tab:notation} summarizes the key mathematical notations used throughout this paper.
\begin{figure*}[htb]
    \centering
    \begin{minipage}{.3\linewidth}
        \centering
        \includegraphics{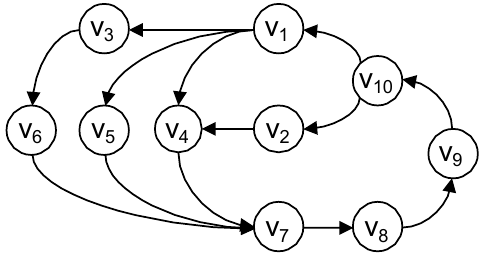}
        \caption{A Directed Graph.}
        \label{fig:ori_g}
    \end{minipage}%
    \begin{minipage}{.7\linewidth}
        \centering
        \captionof{table}{Shortest Path Counting Labels of Fig.~\ref{fig:ori_g}}
        \vspace{-2mm}
        \scalebox{0.8}{
          \begin{tabular}{|l|l|l|}
            \hline
            \cellcolor{gray!25}\textbf{Vertex} & \cellcolor{gray!25}\textbf{$L_{in}(\cdot)$} & \cellcolor{gray!25}\textbf{$L_{out}(\cdot)$}\\
            \hline
            $v_{1}$ & $(v_{1},0,1)$ & $(v_{1},0,1)$\\
            \hline
            $v_{2}$ & $(v_{1},6,2)$ $(v_{7},4,1)$ $(v_{10},1,1)$ $(v_{2},0,1)$ & $(v_{1},6,1)$ $(v_{7},2,1)$ $(v_{4},1,1)$ $(v_{2},0,1)$\\
            \hline
            $v_{3}$ & $(v_{1},1,1)$ $(v_{3},0,1)$ & $(v_{1},6,1)$ $(v_{7},2,1)$ $(v_{3},0,1)$\\
            \hline
            $v_{4}$ & $(v_{1},1,1)$ $(v_{7},5,1)$ $(v_{4},0,1)$ & $(v_{1},5,1)$ $(v_{7},1,1)$ $(v_{4},0,1)$\\
            \hline
            $v_{5}$ & $(v_{1},1,1)$ $(v_{5},0,1)$ & $(v_{1},5,1)$ $(v_{7},1,1)$ $(v_{5},0,1)$\\
            \hline
            $v_{6}$ & $(v_{1},2,1)$ $(v_{3},1,1)$ $(v_{6},0,1)$ & $(v_{1},5,1)$ $(v_{7},1,1)$ $(v_{6},0,1)$\\
            \hline
            $v_{7}$ & $(v_{1},2,2)$ $(v_{7},0,1)$ & $(v_{1},4,1)$ $(v_{7},0,1)$\\
            \hline
            $v_{8}$ & $(v_{1},3,2)$ $(v_{7},1,1)$ $(v_{8},0,1)$ & $(v_{1},3,1)$ $(v_{7},5,1)$ $(v_{4},4,1)$ $(v_{10},2,1)$ $(v_{8},0,1)$\\
            \hline
            $v_{9}$ & $(v_{1},4,2)$ $(v_{7},2,1)$ $(v_{8},1,1)$ $(v_{9},0,1)$ & $(v_{1},2,1)$ $(v_{7},4,1)$ $(v_{4},3,1)$ $(v_{10},1,1)$ $(v_{9},0,1)$\\
            \hline
            $v_{10}$ & $(v_{1},5,2)$ $(v_{7},3,1)$ $(v_{10},0,1)$ & $(v_{1},1,1)$ $(v_{7},3,1)$ $(v_{4},2,1)$ $(v_{10},0,1)$\\
            \hline
        \end{tabular}}
        \label{tab:labels}
    \end{minipage}
\end{figure*}

\vspace{1mm}
\noindent
{\bf Problem Statement.}
Given a graph $G$ and a query vertex $v$, the shortest cycle counting query, denoted by SCCnt($v$), returns the number of shortest cycles through vertex $v$. 

\begin{example}~There are three shortest cycles in Figure~\ref{fig:ori_g} with length $6$ through $v_{7}$.
Thus, SCCnt($v_{7}$)=3.
\end{example}
\subsection{Hub Labeling for Shortest Path Counting Queries}
Next, we present the hub labeling method for the shortest path counting between vertices $s$ and $t$, SPCnt($s,t$). 
In a directed graph, the shortest path counting between vertices $s$ and $t$ seeks to determine the total number of all the shortest paths from $s$ to $t$. \cite{zhang20} proposed a 2-hop labeling scheme and index construction algorithm to build the index efficiently and enabling real-time shortest path counting queries. 
The hub labeling scheme supports the cover constraint, \underline{E}xact \underline{S}hortest \underline{P}ath \underline{C}overing (ESPC), which implies that it not only encodes the shortest distance between two vertices but also ensures that such shortest paths are correctly counted. HP-SPC is their proposed algorithm to construct the SPC label index that satisfies ESPC.

Formally, given a directed graph $G$, HP-SPC assigns each vertex $v \in G$ an in-label $L_{in}(v)$ and an out-label $L_{out}(v)$, which are composed of entries of the form $(w, sd(v,w), \theta_{v,w})$. The shortest distance between $v$ and $w$ is denoted by $sd(v,w)$, and the number of shortest paths between $v$ and $w$ is denoted by $\theta_{v,w}$. If $w \in L_{in}(v)$ or $w \in L_{out}(v)$, we say that $w$ is a hub of $v$. 
Essentially, the in-label $L_{in}(v)$ keeps track of the distance and counting information from its hubs to itself, while the out-label $L_{out}(v)$ records distance and counting information from $v$ to its hubs. HP-SPC adheres to the cover constraint, which states that for any given starting vertex $s$ and ending vertex $t$, there exists a vertex $w \in L_{out}(s) \cap L_{in}(t)$ that lies on the shortest path from $s$ to $t$. 
Additionally, HP-SPC guarantees the correctness of shortest path counting by including the shortest path from $s$ to $t$ via a hub vertex once during the label construction. 
SPCnt($s,t$) is evaluated by scanning the $L_{out}(s)$ and $L_{in}(t)$ for the shortest distance via common hubs and adding the multiplication of the corresponding counting. Equation \eqref{con:spc1} identifies all common hubs (on shortest paths) from $L_{out}(s)$ and $L_{in}(t)$. Equation \eqref{con:spc2} determines the result of SPCnt($s,t$).

\begin{equation}
H = \{h | \mathop{\arg\min}\limits_{h \in L_{out}(s) \cap L_{in}(t)} \{sd(s,h) + sd(h,t)\}\} \label{con:spc1}
\end{equation}
\begin{equation}
SPCnt(s,t) = \sum_{h\in H} \theta(s,h) \cdot \theta(h,t) \label{con:spc2}
\end{equation}

\begin{example}
\autoref{fig:ori_g} depicts a directed graph with 10 vertices, and \autoref{tab:labels} provides its hub labeling index for SPCnt queries.  We use SPCnt($v_{10},v_{8}$) as an example to determine the shortest paths counting from $v_{10}$ to $v_{8}$. Two common hubs \{$v_{1}, v_{7}$\} are discovered by scanning $L_{out}(v_{10})$ and $L_{in}(v_{8})$. The shortest distance through $v_{1}$ is 1 + 3 = 4, while the counting is 1 $\cdot$ 2 = 2; The shortest distance via $v_{7}$ is 3 + 1 = 4, and the counting is 1 $\cdot$ 1 = 1. Thus, the number of shortest paths from $v_{10}$ to $v_{8}$ is 2 + 1 = 3 with length 4. 
\end{example}





\section{Baselines}
\label{sec:baselines}
In this section, we introduce two baseline solutions for the shortest cycle counting problem in this paper.

\subsection{HP-SPC for SCCnt by Neighborhood Information}
\label{sec:espc}
By selecting the identical beginning and ending vertex, an obvious solution for shortest cycle counting is to utilize the existing hub-labeling techniques HP-SPC 
for shortest path counting. To compute  SCCnt($v_q$) through query vertex $v_q$, we resort to HP-SPC 
based hub labeling described in \cite{zhang20} by specifying both starting and ending vertex to be $v_q$, i.e., SPCnt($v_q, v_q$). Nevertheless, this will lead to 
erroneous query results since the shortest distance is always 0, as determined by identifying a self-loop of $v_q$ in \eqref{con:spc1}. For instance, if we aim to count the number of shortest cycles via $v_1$ in \autoref{fig:ori_g}, SPCnt($v_1, v_1$) will search up $L_{in} (v_1)$ and $L_{out} (v_1)$ of the hub labeling in \autoref{tab:labels}, and return 0. 

To remedy this, we convert the shortest cycle counting  query to the shortest path counting query between query vertex $v_q$ and the in-neighbors (or out-neighbors) of $v_q$. Formally, to compute SCCnt($v_q$), we compute the SPCnt from $v_q$  to each of the in-neighbor $w$ of $v_q$, or from each out-neighbor to $v_q$. To reduce query costs, we choose out-neighbors of $v_q$ when $|nbr_{out}(v_q)| < |nbr_{in}(v_q)|$, and in-neighbors of $v_q$ otherwise.  
Equation \eqref{con:scc1} and \eqref{con:scc2} demonstrate the evaluation of SCCnt($v_q$) when $|nbr_{out}(v_q)| < |nbr_{in}(v_q)|$. The first step is to get all vertices with the shortest distance to $v_q$ using Equation \eqref{con:scc1}. Clearly, the length of the shortest cycles via $v_q$ is increasing the shortest distance in \eqref{con:scc1} by 1. The number of shortest cycles passing through $v_q$ is derived by summing up the number of shortest paths between $v_q$ and each of the vertices acquired in the first step. Notably, if there is no path from $v_q$ to its neighbors, Equation \eqref{con:scc1} returns an empty set, indicating that there is no cycle via $v_q$. 
\begin{equation}
W = \{ w | \mathop{\text{argmin}}_{w \in nbr_{out}(v_q)} sd(w,v_q) \} 
\label{con:scc1}
\end{equation}

\begin{equation}
SCCnt(v_q) = \sum_{w\in W} SPCnt(w,v_q) \label{con:scc2}
\end{equation}

\begin{example}
Consider the example of SCCnt($v_{7}$). $v_{7}$ has three in-neighbors $\{v_{4}, v_{5}, v_{6}\}$. According to the index in Table II, SPCnt($v_{7}, v_{4}$) = 2 and $sd(v_{7}, v_{4})$ = 5; SPCnt($v_{7}, v_{5}$) = 1 and $sd(v_{7}, v_{5})$ = 5; SPCnt($v_{7}, v_{6}$) = 1 and $sd(v_{7}, v_{6})$ = 6. Thus, the shortest cycles via $v_{7}$ passing through $v_{4}$ and $v_{5}$. SCCnt($v_{7}$) is 2 + 1 = 3. 
\end{example}
\eatYou{Clearly, the complexity of the SPC-based solution is proportional to the number of $v_q$'s neighbors. SCCnt($v_q$) may still be prohibitively expensive for vertices with an excessive number of neighbors. }


\subsection{Breadth-First Search (BFS) Based Shortest Cycle Counting } 
The second baseline solution applies the Breadth-First Search (BFS) from the query vertex $v_q$. In Algorithm 1, we record the shortest distance from $v_{q}$ ($D[\cdot]$) and the counting of the shortest paths ($C[\cdot]$) for each accessed vertex throughout the BFS. Lines 11-16 compare the recorded distance of accessed vertex $w_n$ and the derived distance from its ancestor $w$. Particularly, if the recorded distance is larger, $D[w_n]$ will be adjusted to reflect $D[w]$ and will inherit the shortest path counting from $w$ (lines 12-14). If the recorded distance equals the distance from $w$, only the shortest path counting will be updated (lines 15-16). Once $v_q$ is popped out from the queue, all of its shortest cycles are encountered. Otherwise, if $v_q$ is never visited until the queue is empty, then it indicates that no cycle through $v_q$ occurs.  
Clearly, the time complexity of BFS based solution is $O(m+n)$ and the space complexity is also $O(m+n)$.

\eat{
\begin{algorithm}
\setstretch{0.95}
    \For{{\bf each} $v \in V$}{
        $D[v] \leftarrow \infty$; $C[v] \leftarrow 0$\;
    }
    $Q \leftarrow  \emptyset$;\\
    \For{{\bf each} $u \in nbr_{out}(v_q)$}{
        $D[u] \leftarrow 1$; $C[u] \leftarrow 1$\;
        $Q.{\rm enqueue}(u)$\;
    }
    \While{$Q$ {\rm is not empty}}{
         $w \leftarrow Q.{\rm dequeue}()$\;
         \If{$w = v_{q}$}{{\bf return} $(D[v_{q}],C[v_{q}])$\;}
         \For{{\bf each} $w_{n} \in nbr_{out}(w)$}{
            \If{$D[w_{n}] > D[w] + 1$}{
                $D[w_{n}] \leftarrow D[w] + 1$; $C[w_{n}] \leftarrow C[w]$\;
                $Q.{\rm enqueue}(w_{n})$\;
            } \ElseIf{$D[w_{n}] = D[w] + 1$}{
                $C[w_{n}] \leftarrow C[w_{n}] + C[w]$\;
            }
        }
    }
    {\bf return} $(\infty,0)$\;
\caption{BFS-CYCLE($G, v_{q}$)}
\label{alg:BFS}
\end{algorithm}
}

\section{Bipartite Hub Labeling for Shortest Cycle Counting}
\label{sect:ours}
This section introduces our bipartite conversion based hub labeling scheme, then describes labeling computation and processing shortest cycle counting based on the labeling. 
\subsection{Bipartite Hub Labeling Scheme}
We seek to develop a hub labeling scheme for shortest cycle counting by effectively implementing hub labeling for shortest path counting. To this end, the labeling must adhere to the cover constraint of hub labeling for shortest path counting. The following constraint is imposed in \cite{zhang20} to ensure correctly maintaining the counting of shortest paths. 

\vspace{1mm}
{\bf{Cover Constraint.}}
Given a total ordering $\prec$ over all the vertices in graph $G$ where for any two vertices ($v$,$w$), $v \prec w$ indicates that $v$ has a rank higher than $w$. A label entry $(v,d,c)$ is added to the in-label of vertex $w$ $L_{in}(w)$ if there exists at least one shortest path $sp(v,w)$ between $v$ and $w$, where $v$ has the highest rank along $sp(v,w)$. Here, $d$ represents the shortest distance $sd(v,w)$, and $c$ is the number of all such shortest paths. Notably, if $c = |$SP$(v,w)|$, namely, $c$ is the number of all shortest paths between $v$ and $w$, $(v,d,c)$ is referred to as a {\em \bf{canonical label}} of $L_{in}(w)$. Otherwise, if $c < |$SP$(v,w)|$, namely only a proper subset of $SP(v, w)$ is counted, $(v,d,c)$ it is a {\em \bf{non-canonical label}}. 

\begin{example}
In \autoref{fig:ori_g}, assume the ordering of vertices follows the degree order: $v_{1}\prec v_{7}\prec v_{4}\prec v_{10}\prec v_{2}\prec v_{3}\prec v_{5}\prec v_{6}\prec v_{8}\prec v_{9}$. There are two shortest paths in the reverse graph (a reverse graph is obtained by reversing the orientations of all edges in the graph) from $v_{4}$ to $v_{10}$ with length $2$. One of the paths $v_4, v_1, v_{10}$ passes through $v_{1}$ which has a higher rank than $v_{4}$. Thus, the count $c$ is $1$ rather than $2$ for $(v_{4},2,c)$ in $L_{out}(v_{10})$. This label entry is an example of a non-canonical one. 
\end{example}

\eatYou{
The out-labels $L_{out}$ for each vertex in $G$ satisfy the same cover constraint and can be derived by examining the reverse graph of $G$. While canonical labeling is adequate for shortest distance computation, the supplement of non-canonical labels is essential for shortest distance counting. This is because such a cover constraint guarantees the exact shortest path counting as each shortest path is counted precisely once.}
 
\subsection{Bipartite Conversion}
It is insufficient to provide shortest cycle counting just by following the cover constraint of shortest path counting. As shown in Section \ref{sec:espc}, current hub labeling techniques are incapable of supporting shortest cycle counting.
We offer a new bipartite conversion approach to circumvent this limitation. 
Given a directed graph $G$, the following procedure is used to build its bipartite graph $G_b$ conversion.

Algorithm~\ref{alg:BI-G} illustrates the conversion. From line~\ref{line:bi_start} to~\ref{line:bi_end}, each vertex $v \in G$ is decomposed into a pair of {\em couple vertices} ($v^{i}$,$v^{0}$) where $v^{i}$  represents the incoming vertex of $v$, and $v^{o}$ denotes the outgoing vertex of $v$, along with an edge $e(v^{i}, v^{o})$. $V_{in}$ ($V_{out}$) denotes the set of all incoming (outgoing) vertices. The edge $e(v,w)$ from the original graph $G$ is replaced by an edge $e(v^{o},w^{i})$ in the new graph $G_{b}$. $G_b$ is a bipartite graph with two vertex sets $V_{in}$ and $V_{out}$. $G_{b}$ has a vertex and edge count of $2\cdot n$ and $n + m$, respectively. The vertices in $V_{in}$ contain all the in-edges of $G$, and the vertices in $V_{out}$ contain all the out-edges. \autoref{fig:bi_g} illustrates the bipartite conversion of the graph in Figure 2. 

\eatYou{
\vspace{1mm}
{\bf Vertex Ordering. }
Vertex ordering is extensively researched in the literature for hub labeling, where degree-based ordering is frequently adopted and results in the state-of-the-art canonical hub labeling for shortest distance queries \cite{akiba13, zhang20}. In this paper, we rank vertices based on the multiplication of in-degree and out-degree of vertices. Each pair of the couple vertices $v^i, v^o$ in the bipartite graph $G_b$ is given consecutive ranks, with $v^{i}$ always having a higher rank than $v^{o}$. In comparison to ranking the couple vertices individually, this approach enables time efficiency for hub labeling index construction and space efficiency of index size. Since there is no gap between the order of couple vertices $v^i$ and $v^o$ and just one edge between them, a vertex that is a hub of $v^i$ ($v^o$) will likewise be a hub for $v^o$ ($v^i$). This property allows the following optimization for hub labeling construction.}

\begin{figure}
  \centering
  \includegraphics[scale=0.9]{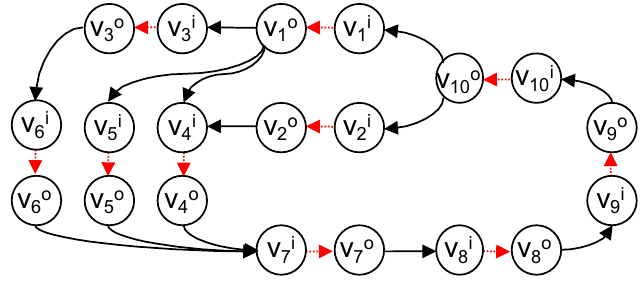}
  \vspace{-3mm}
  \caption{Bipartite Conversion of \autoref{fig:ori_g}}
  \vspace{-3mm}
  \label{fig:bi_g}
\end{figure}

\begin{algorithm}
\setstretch{0.95}
    $V_{in} \leftarrow \emptyset; V_{out} \leftarrow \emptyset; E_{b} \leftarrow \emptyset$\;
    \For{{\bf{each}} $v \in V$}{
    \label{line:bi_start}
        $\{v^{i}, v^{o}\} \leftarrow v$\;
        $V_{in} \leftarrow V_{in} \cup \{v^{i}\}$\; 
        $V_{out} \leftarrow V_{out} \cup \{v^{o}\}$\;
        $E_{b} \leftarrow E_{b} \cup \{e(v^{i}, v^{o})\}$\;
        \label{line:bi_end}
    }
    \For{{\bf{each}} $e(v,w) \in E$}{
        $E_{b} \leftarrow E_{b} \cup e(v^{o}, w^{i})$\;
    }
    {\bf{return}} $G_{b}=(V_{in}, V_{out}, E_{b})$\;
\caption{\textsc{Bi}-G($G$)}
\label{alg:BI-G}
\end{algorithm}

\vspace{1mm}
{\bf Couple-Vertex Skipping}.
The consecutive order of each pair of couple vertices allows the correct hub labeling by avoiding some computation associated with index construction. According to the cover constraint, only higher-ranked vertices are eligible to be the hub of lower-ranked ones. Thus, $(v^{i}, 1, 1)$ must be a label entry of $L_{in}(v^{o})$. Any vertex that is qualified to serve as an in-label hub for $v^{i}$ should also be an in-label hub of $v^{o}$. When an in-label hub is identified for $v^{i}$, we can simply increase the distance by 1 and add the label to $v^{o}$. As a result, we can exclude $v^{o}$ from the in-label construction. Regarding out-labels, the graph is traversed in a reverse direction. Similarly, if an out-label hub is identified for $v^{o}$, it will be immediately added to $v^{i}$, obviating the need for $v^i$'s out-label construction.  
Additionally, if both coupled vertices $v^i$ and $v^o$ serve as hubs for a vertex $w$, it is acceptable to retain just the label entry of the higher-ranked hub $v^i$, since $v^o$ contains redundant information from the pair.

\subsection{Bipartite Hub Labeling Construction}
The index construction algorithm for in-label generation is illustrated in Algorithm \ref{alg:SCC}, which employs the couple-vertex skipping technique. The algorithm begins with the vertex $v$ and pushes it to the vertices that share $v$ as a hub. It explores from higher to lower rank in accordance with the vertex ordering.

\begin{algorithm}[htb]
\setstretch{0.95}
    \For{{\bf each} $v \in V_{in}\cup V_{out}$}{
        $L^{c}_{in}(v) \leftarrow \emptyset$; $L^{nc}_{in}(v) \leftarrow \emptyset$\;
        $L^{c}_{out}(v) \leftarrow \emptyset$; $L^{nc}_{out}(v) \leftarrow \emptyset$\;
        $D[v] \leftarrow \infty$; $C[v] \leftarrow 0$\;
    }
    \For{{\bf each} $v \in V_{in}\cup V_{out}$ {\rm in descending order}}{
        \If{$v \in V_{out}$}{
            append ($v,0,1$) to $L^{c}_{in}(v)$; append ($v,0,1$) to $L^{c}_{out}(v)$\;
            {\bf continue}\;
        }
        Queue $Q \leftarrow $ $\emptyset$; $Q.{\rm enqueue}(v)$\;
        $D[v] \leftarrow 0$; $C[v] \leftarrow 1$; $V_{is} \leftarrow \{v\}$\;
        \While{$Q$ {\rm $\neq \emptyset$}}{
            $w \leftarrow Q.{\rm dequeue}()$\;
            $d \leftarrow min_{u \in L^{c}_{out}(v) \cap L^{c}_{in}(w)}sd(v,u)+sd(u,w)$\;
            \If{$d < D[w]$}{
                {\bf continue}\;
            }
            \textsc{Insert}LABEL($v,d,w,D[w],C[w]$)\;
                $w' \leftarrow $ couple of $w$; $V_{is} \leftarrow V_{is} \cup \{w'\}$\;
                $D[w'] \leftarrow D[w] + 1$; $C[w'] \leftarrow C[w]$\;
            \For{{\bf each} $w_{n} \in nbr_{out}(w')$}{
                \If{$D[w_{n}] = \infty \land v \prec w_{n}$}{
                    $D[w_{n}] \leftarrow D[w'] + 1$; $C[w_{n}] \leftarrow C[w']$\;
                    $Q.{\rm enqueue}(w_{n})$; $V_{is} \leftarrow V_{is} \cup \{w_{n}\}$\;
                } \ElseIf{$D[w_{n}] = D[w'] + 1$}{
                    $C[w_{n}] \leftarrow C[w_{n}] + C[w']$\;
                }
            }
        }
        \For{{\bf each} $v \in V_{is}$}{
            $D[v] \leftarrow \infty$; $C[v] \leftarrow 0$\;
        }
        
        {/* Out-labels Generation */}
    }
    \For{{\bf each} $v \in V_{in}\cup V_{out}$}{
        $L_{in}(v) \leftarrow L^{c}_{in}(v) \cup L^{nc}_{in}(v)$\;
        $L_{out}(v) \leftarrow L^{c}_{out}(v) \cup L^{nc}_{out}(v)$\;
    }
\caption{CSC($G_{b},\bar{G_{b}}$)}
\label{alg:SCC}
\vspace{-1mm}
\end{algorithm}

\begin{figure*}[htb]
	\begin{center}
		\subfigure[]{
			\label{in_1}
			\centering
			\includegraphics[scale=0.22]{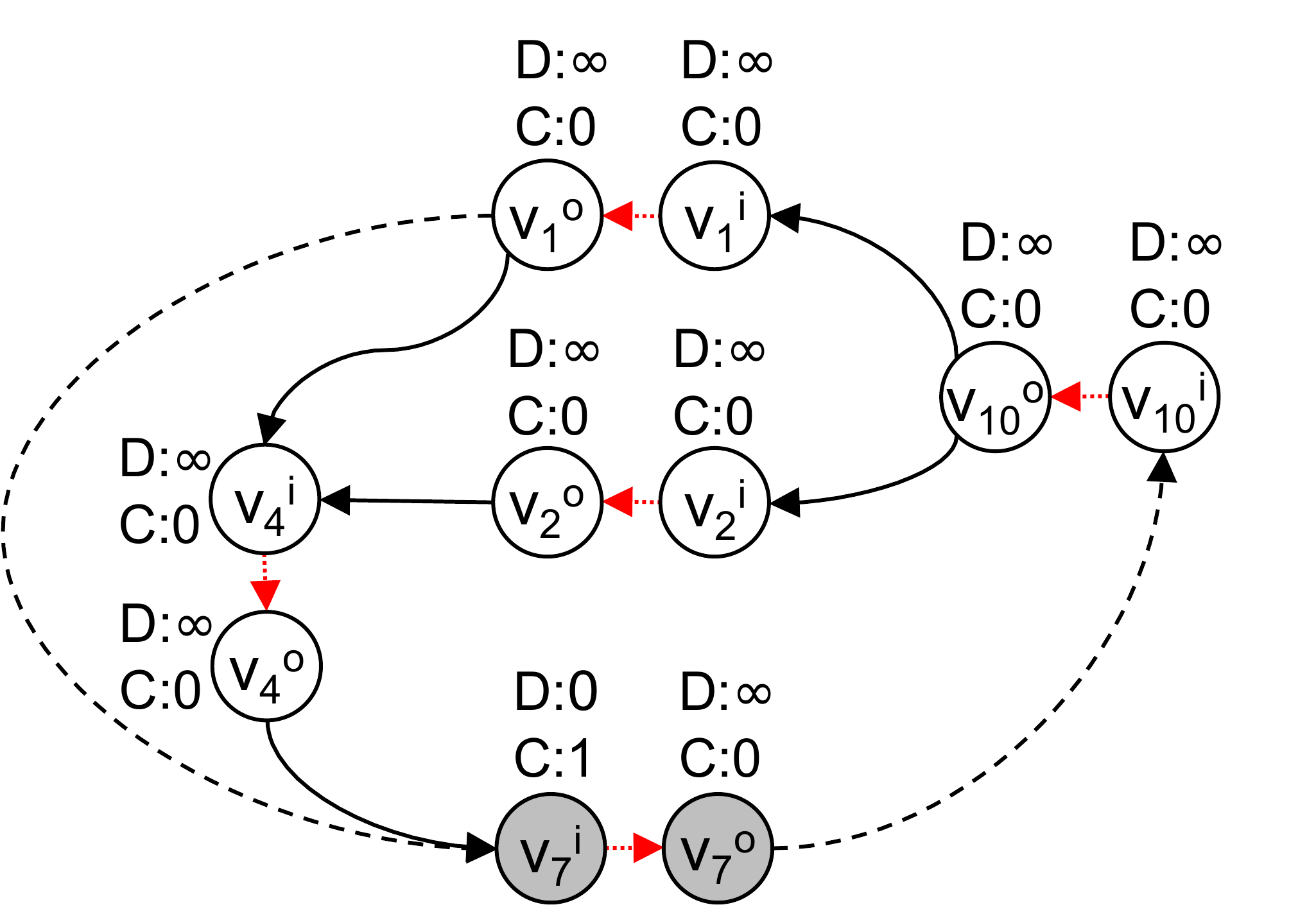} 
		}\hspace{2mm}
		\subfigure[]{
			\label{in_2}
			\centering
			\includegraphics[scale=0.22]{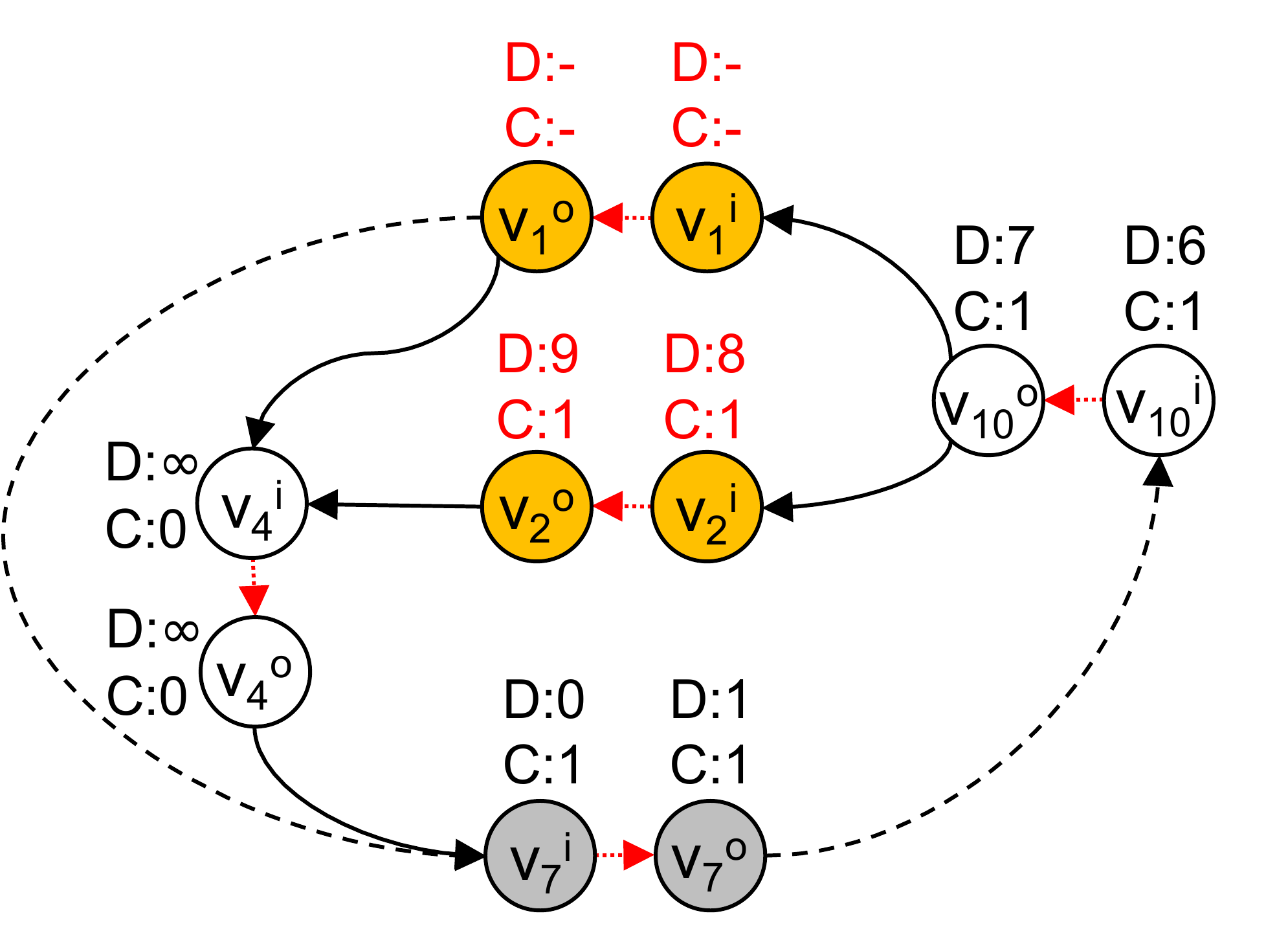}     
		}\hspace{2mm}
		\subfigure[]{
			\label{in_3}
			\centering
			\includegraphics[scale=0.22]{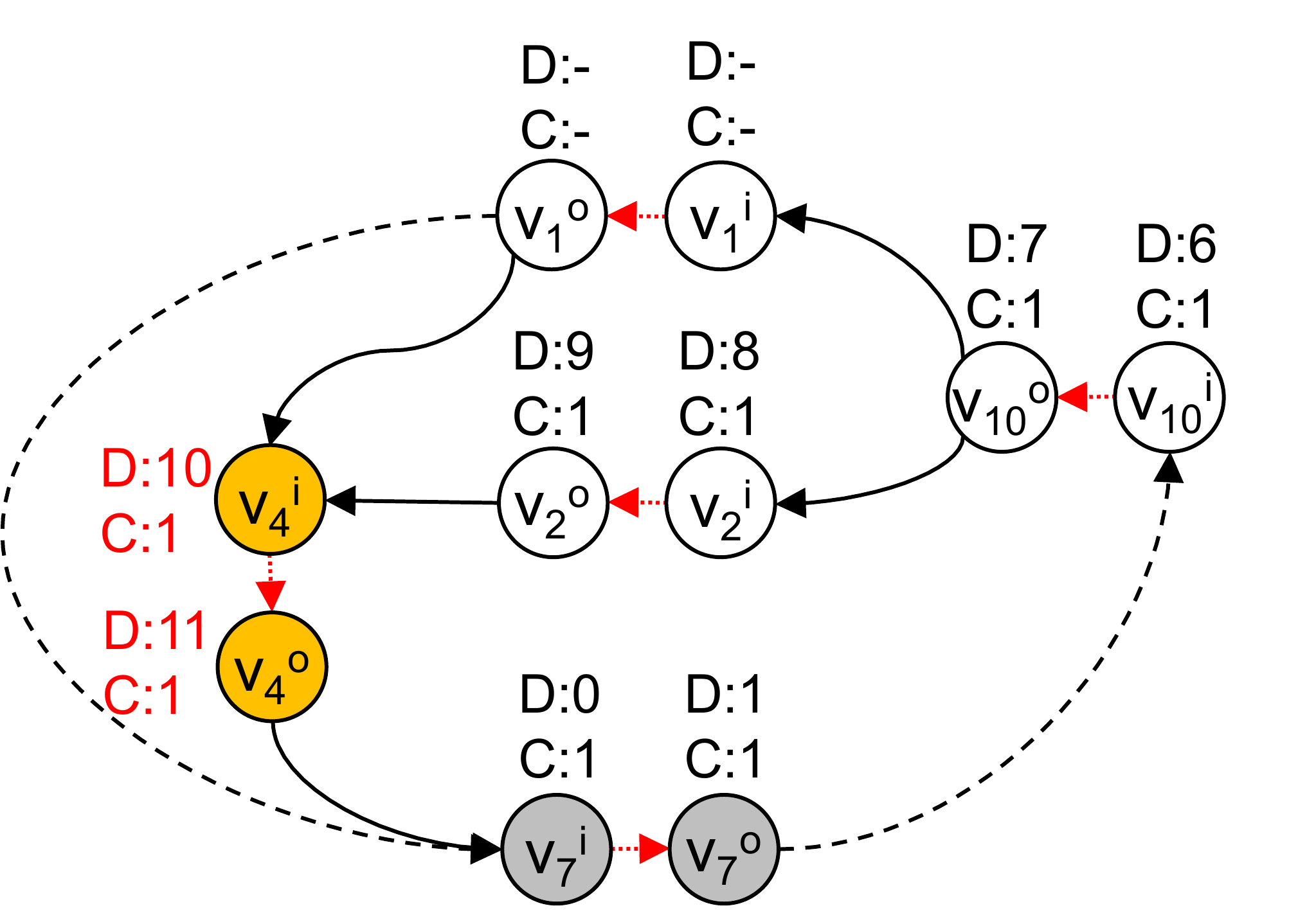}     
		}\hspace{2mm}
	\end{center}
	\vspace{-5mm}
	\caption{In-Labels Construction Example (Hub ${v_{7}}^{i}$).} \label{fig:ex_index_in}
	\vspace{-5mm}
\end{figure*}

\eat{\begin{figure*}[htb]
	\begin{center}
		\subfigure[]{
			\label{out_1}
			\centering
			\includegraphics[scale=0.22]{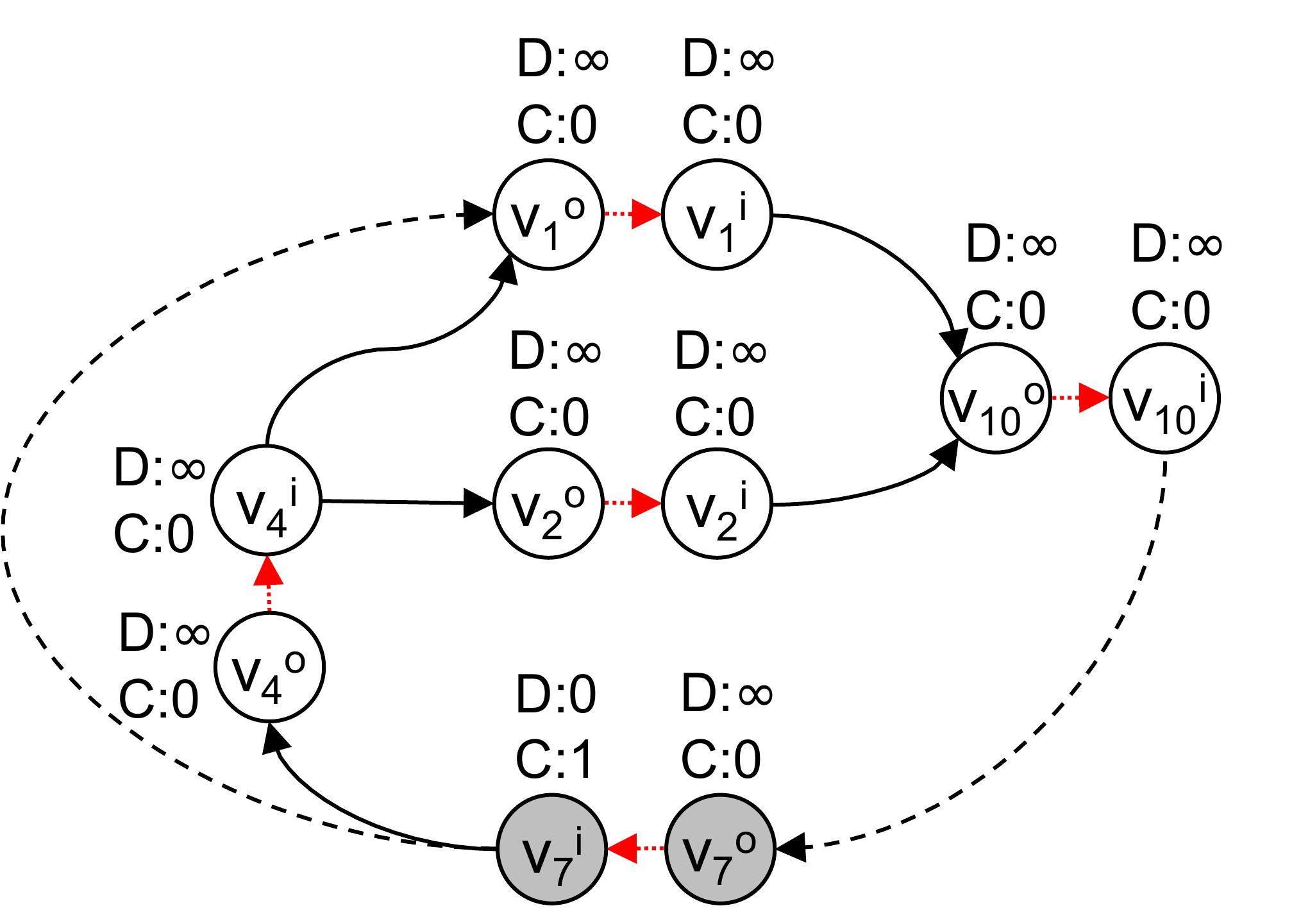}
		}\hspace{2mm}
		\subfigure[]{
			\label{out_2}
			\centering
			\includegraphics[scale=0.22]{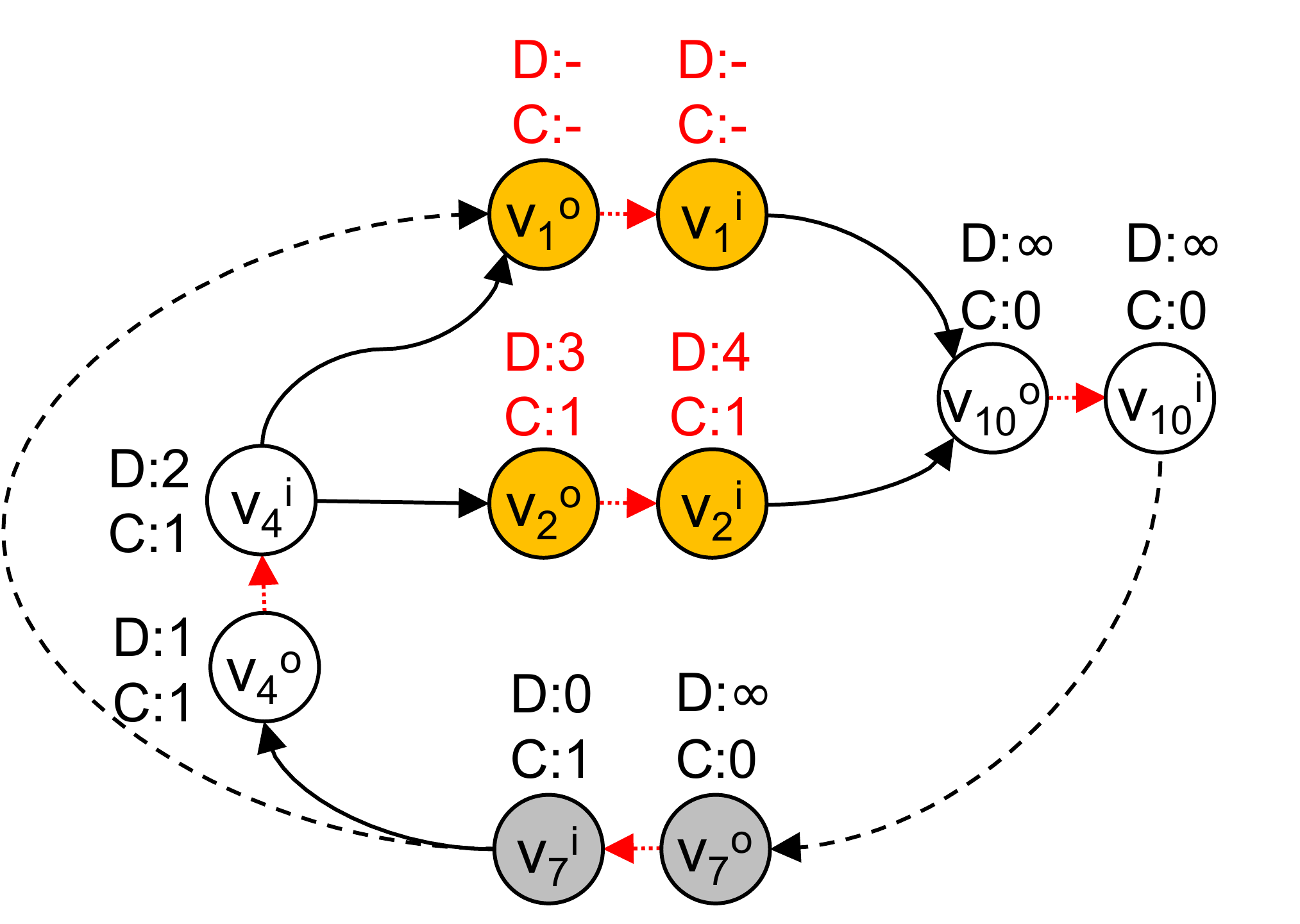}
		}\hspace{2mm}
		\subfigure[]{
			\label{out_3}
			\centering
			\includegraphics[scale=0.22]{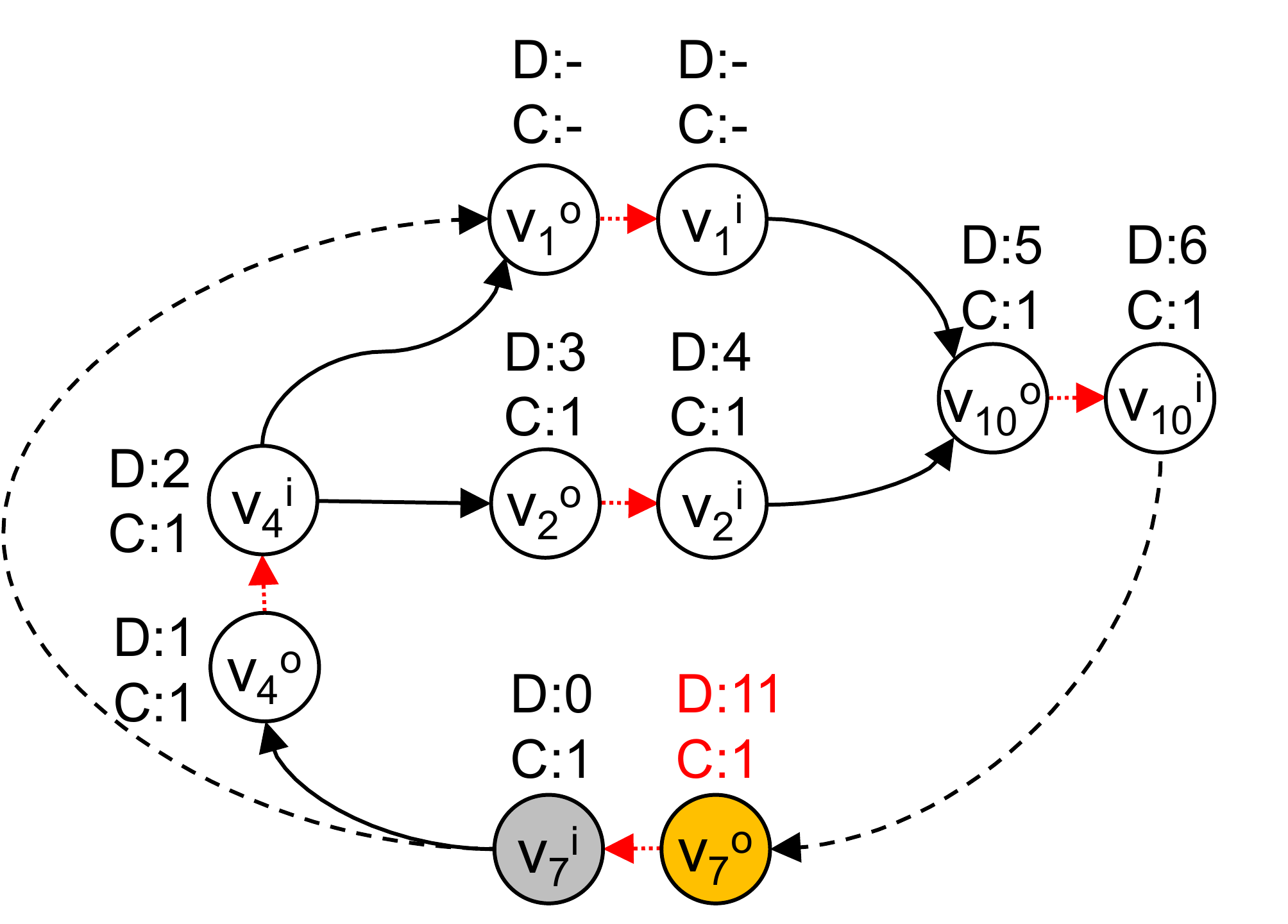} 
		}\hspace{2mm}
	\end{center}
	\vspace{-5mm}
	\caption{Out-Labels Construction Example (Hub ${v_{7}}^{i}$).} \label{fig:ex_index_out}
	\vspace{-5mm}
\end{figure*}}

\vspace{1mm}
{\bf Initialization. } To begin, the canonical and non-canonical in(out)-label sets $(L^{c}_{in}(\cdot), L^{nc}_{in}(\cdot),L^{c}_{out}(\cdot),L^{nc}_{out}(\cdot))$ for each vertex are initialized (lines 2-3). Arrays $D[\cdot]$ and $C[\cdot]$ store tentative distances and counting of such distances (line 4). Then, for each vertex $v$ in descending order of rank, a breadth-first search is performed to identify the vertices $w$ that have $v$ as a hub in $L_{in}(w)$ or $L_{out}(w)$. If $v$ comes from $V_{out}$, we can safely skip the current loop and add the label entries for itself (lines 6-8). If $v$ is from $V_{in}$, the label construction process begins with a queue $Q$ containing the current hub. $Vis$ maintains the set of visited vertices.



{\bf Label Generation.} The following procedure generates all label entries for the current hub $v$ (lines 12-16). From $v$, we visit vertex $w$ in a BFS fashion. We calculate the distance from $v$ to $w$ using the hub labels prior to $v$. This is the shortest distance $d$ in line 13 from $v$ to $w$ through a higher-ranked vertex. If $d$ is less than the tentative distance $D[w]$, the BFS is superfluous, since $v$ is not the highest-ranked vertex along any shortest path from $v$ to $w$ (lines 14-15). Otherwise, a new label of $w$ with hub $v$ is generated by calling Algorithm \ref{alg:IST-SCCnt}. If $d$ equals to the tentative distance $D$, then $v$ has the highest rank along \textbf{some} of the shortest paths from $v$ to $w$. A non-canonical label is generated in this scenario (lines 2-3). According to the Couple-Vertex Skipping, $w'$, the couple of $w$, also has a new non-canonical label with the distance of $D+1$ (line 4). If $d > D$, then $v$ has the highest rank along \textbf{all} the shortest paths from $v$ to $w$. In this case, a canonical label is generated for both $w$ and its couple $w'$ (lines 5-7).

\begin{algorithm}[htb]
        $w' \leftarrow $ couple of $w$\;
    \If{$d = D$}{
        append $(v,D,C)$ to $L^{nc}_{in}(w)$\;
            append $(v,D + 1,C)$ to $L^{nc}_{in}(w')$\;
    } \ElseIf{$d > D$}{
        append $(v,D,C)$ to $L^{c}_{in}(w)$\;
            append $(v,D + 1,C)$ to $L^{c}_{in}(w')$\;
    }
\caption{\textsc{Insert}LABEL($v,d,w,D,C$)}
\label{alg:IST-SCCnt}
\vspace{-1mm}
\end{algorithm}


{\bf Update Tentative Distance and Counting.}~$w \in V_{in}$ is constantly established during the in-label generation process. Due to the fact that the first vertex in $Q$ is from $V_{in}$, then we have its couple $w'$ in $V_{out}$ (line 17). By using the couple-vertex skipping technique, $w$ is skipped. We continue to check each neighbor $w_{n}$ from $nbr_{out}(w')$, where $w_{n}$ must be from $V_{in}$ (lines 19-24). If $w_{n}$ is not visited and has a lower rank than the current hub $v$, its tentative distance and counting are updated. After that, it is pushed into $Q$ (lines 20-22). Thus, only those vertices from $V_{in}$ are eligible to be included in $Q$. Otherwise, if the next-step distance to $w_{n}$ is identical to its tentative distance, the counting is accumulated due to the new shortest paths (lines 23-24). Finally, the label of each vertex is the concentration of the canonical and non-canonical labels (lines 28-30).

The process for out-label generation is similar to the procedure described in lines 9-26 but in a reverse direction. The primary distinctions are as follows: (1) The distance calculated using an existing index should be the distance from $w$ to $v$ (line 13); (2) Replace $nbr_{out}$ with $nbr_{in}$ (line 19); (3) The first time $w$ popped from $Q$ is $v$ (line 12), only ($v,0,1$) should be inserted into $L^{c}_{out}(v)$ in this loop. Then skip lines 17-18 and verify each of $v$'s in-neighbor. From the second loop, $w$ should always be from $V_{out}$; (4) In Algorithm~\ref{alg:IST-SCCnt}, label entries should be inserted into out-labels. If $w$ is the couple of $v$, only insert label entry into $L_{out}(w)$ and prune at this point. The Algorithm CSC takes the graph $G_{b}$ and reverse graph $\bar G_{b}$ and generates the in-label $L_{in}(\cdot)$ and out-label $L_{out}(\cdot)$ set for each vertex $v$.

\begin{example}
\label{case:in}
\autoref{fig:ex_index_in} shows three snapshots during in-label construction for hub ${v_{7}}^{i}$. These snapshots that correspond to \autoref{fig:bi_g}, except that $v_3$, $v_5$, $v_6$, $v_8$ and $v_9$ are omitted for ease. The initial stages are shown in \autoref{fig:ex_index_in}(a). Except for ${v_{7}}^{i}$'s counting, all tentative distance and counting are set to $\infty$ and $0$, respectively. After reaching ${v_{10}}^{o}$ with distance 7 and counting 1, ${v_{1}}^{i}$ and ${v_{2}}^{i}$ will be processed the next round in \autoref{fig:ex_index_in}(b). Due to the fact that ${v_{1}}^{i} \prec {v_{7}}^{i}$, the BFS at $v_1$ is trimmed, and the algorithm proceeds to ${v_{2}}^{i}$. All the in-label entries with hub ${v_{7}}^{i}$ are canonical prior to ${v_{4}}^{i}$. In \autoref{fig:ex_index_in}(c), we can see that when we reach ${v_{4}}^{i}$, $sd({v_{7}}^{i}, {v_{4}}^{i})$ is 10 via hub ${v_{1}}^{i}$ which is the current tentative distance. Thus, $({v_{7}}^{i},10,1)$ is inserted into $L^{nc}_{in}({v_{4}}^{i})$ as a non-canonical label entry.
\end{example}

As with Example~\ref{case:in}, \autoref{fig:ex_index_out} illustrates the generation of out-labels with hub ${v_{7}}^{i}$.


\subsection{Query Evaluation}
Counting the shortest cycles through $v$ in the original graph $G_{0}$ corresponds to counting the shortest paths from $v^{o}$ to $v^{i}$ in the transformed bipartite graph $G_{b}$. SCCnt($v$) can be assessed using the index constructed above as SPCnt($v^{o}, v^{i}$). 
Note that the distance $d$ in the result represents the shortest distance from $v^{o}$ to $v^{i}$ in $G_{b}$. Due to the presence of an edge $e(v^{i},v^{o})$ and the doubling of all distances in $G_{b}$, the shortest cycle length in $G_{0}$ is ($d$ + 1) / 2. 

\begin{example}
Consider SCCnt($v_{7}$) in comparison to Example 4. \autoref{tab:v_7_label} shows the in-labels for $v_{7}^{i}$ and out-labels for $v_{7}^{o}$. The shortest distance from $v_{7}^{o}$ to $v_{7}^{i}$ through hub $v_{1}^{i}$ is 4 + 7 = 11 and the associated path counting is 2 $\cdot$ 1 = 2. The shortest path counting with the same length via hub $v_{7}^{i}$ is 1 $\cdot$ 1 = 1. Thus, SCCnt($v_{7}$) is 2 + 1 = 3. The shortest cycle length is (11 + 1) / 2 = 6. The evaluation of SCCnt($v$) is evaluated only in terms of $L_{in}(v_{7}^{i})$ and $L_{out}(v_{7}^{o})$.
\end{example}

\vspace{-5mm}
\begin{table}[H]
\centering
\caption{In-Labels of $v_{7}^{i}$ and Out-Labels of $v_{7}^{o}$}
\vspace{-2mm}
  \begin{tabular}{|l|l|}
    \hline
    $L_{in}(v_{7}^{i})$ & $(v_{1}^{i},4,2)$ $(v_{7}^{i},0,1)$\\
    \hline
    $L_{out}(v_{7}^{o})$ & $(v_{1}^{i},7,1)$ $(v_{7}^{i},11,1)$ $(v_{7}^{o},0,1)$\\
  \hline
\end{tabular}
\label{tab:v_7_label}
\vspace{-3mm}
\end{table}

\subsection{Index Reduction}
Each couple pair's sequential order ensures that there is no rank gap between them. Nonetheless, $v^{i}$ is capable of serving as a hub for $v^{o}$'s out-labels. It relies on whether there exists $sp(v^{o},v^{i})$ with $v^{i}$ as the highest rank vertex along the path. As a result, we can keep a duplicate of each pair of couple vertices' in-labels (out-labels). When the complete index must be recovered, we just need to modify the distance element and the $v^{i}$-hub out-label entry for $v^{o}$ if necessary.

\subsection{Complexity Analysis}

\begin{theorem}
For a graph with treewidth $\omega$, average out-degree $s_f$ and average in-degree $s_b$, the index construction time of CSC is O($n\omega^{2}\log^{2}n + (s_f + s_b)\cdot n\omega \log n$), and its index size is O($n\omega \log n$).
\label{the:complexity}
\end{theorem}

\begin{proof}
Assume that $G_{b}$ has a centroid decomposition $(X,T)$. $T$ is a tree in which each node $t$ corresponds to a subset of $V_{in}\cup V_{out}$ termed a bag $X_{t}$ and $X=\{X_{t}$ $|$ $t\in T\}$. Due to the fact that we use the couple-vertex skipping technique, each pair of couple vertices can be considered as a single vertex. We assume that the number of vertices is $n$ which is the same as that of $G_{0}$. Vertices in the root bag are ranked highest, followed by vertices in the root's successors. The BFS will never visit the vertices beyond the current vertex (hub) depth due to the labeling constraint and the ordering of label creation. A bag contains at most $\omega$ vertices, and the depth of $T$ is $\log n$. Thus, each vertex has $2\omega \log n$ label entries, for a total of O$(n\omega \log n)$. The time complexity of evaluating a query is $O(\omega \log n)$.

Let $s_f$ represent the average number of successors and $s_b$ represent the average number of ancestors. Each time a label entry is produced, we traverse $s_f$ ($s_b$) edges in a forward (backward) direction, respectively. In $O(\omega \log n)$ time, a query can be evaluated. Thus, the total time required to build an index is $O((s_f + \omega\log n) \cdot n\omega \log n + (s_b + \omega \log n)\cdot n \omega \log n)$ which is equal to $O( n \omega^{2} \log^{2}n + (s_f + s_b) \cdot n \omega \log n)$.
\end{proof}

\section{Maintenance over Dynamic Graphs}
\label{sect:insert}
This section proposes an efficient algorithm for updating the bipartite hub labeling in the presence of edge insertions and deletions. 
Note that vertices updates can be accomplished via a series of edges updates. Thus, we concentrate on edge updates. 


\subsection{Edge Insertion} 
The index can be updated by simply reconstructing the whole index. To prevent performing actions on labels that are not impacted, we specify the affected hubs. The affected hubs could be out-of-date or inserted as new label entries. Our update algorithm's central idea is to perform pruned BFSs beginning from these affected hubs. The following are related lemmas and theorems:

\begin{lemma}
The shortest distance between any pair of vertices does not grow with the addition of a new edge.
\label{no_incdis}
\end{lemma}


\begin{lemma}
If the shortest distance from $v$ to $w$ changes as a result of the insertion of edge $e(a,b)$, then all the new shortest paths from $v$ to $w$ travel via $e(a,b)$.
\label{new_short}
\end{lemma}


According to Lemma \ref{no_incdis} and Lemma \ref{new_short}~\cite{akiba14}, some label entries should be updated or generated if the insertion of ($a,b$) results in new shortest paths. To discover new shortest paths, suppose we perform BFS in forward direction from a certain $v_{k}$, where $sd(v_{k},a)<sd(v_{k},b)$, and the BFS begins from $b$ with distance $sd(v_{k},a)+1$ and appropriate shortest path counting information. For the reverse direction search from a specific $u_{k}$, where $sd(b,u_{k})<sd(a,u_{k})$, the BFS should begin at $a$ with distance $sd(b,u_{k})+1$. Due to the fact that we need to update the index with all the new shortest paths, the BFS from $v_{k}$ in the forward direction should be processed as follows.

When a vertex $w$ is encountered, let $D$ represent the tentative distance from $v_{k}$ to $w$ and $D_{L}$ denote the distance calculated using the current index. Additionally, the provisional shortest path counting is recorded. There are three possible scenarios:
\begin{enumerate}[leftmargin=32pt]
    \item[$Case_1$] If $D > D_{L}$, the BFS is pruned because the shortest path from $v_{k}$ to $w$ does not pass through $e(a,b)$;
    \item[$Case_2$] If $D = D_{L}$, we accumulate the counting as new same-length shortest path is discovered, and continue BFS;
    \item[$Case_3$] If $D < D_{L}$, we update the distance and counting as new shortest path with shorter distance is discovered, and continue BFS.
\end{enumerate}

\eat{\begin{figure}[t]
	\begin{center}
		\subfigure[A graph]{
			\label{inc_1}
			\centering
			\includegraphics[scale=0.25]{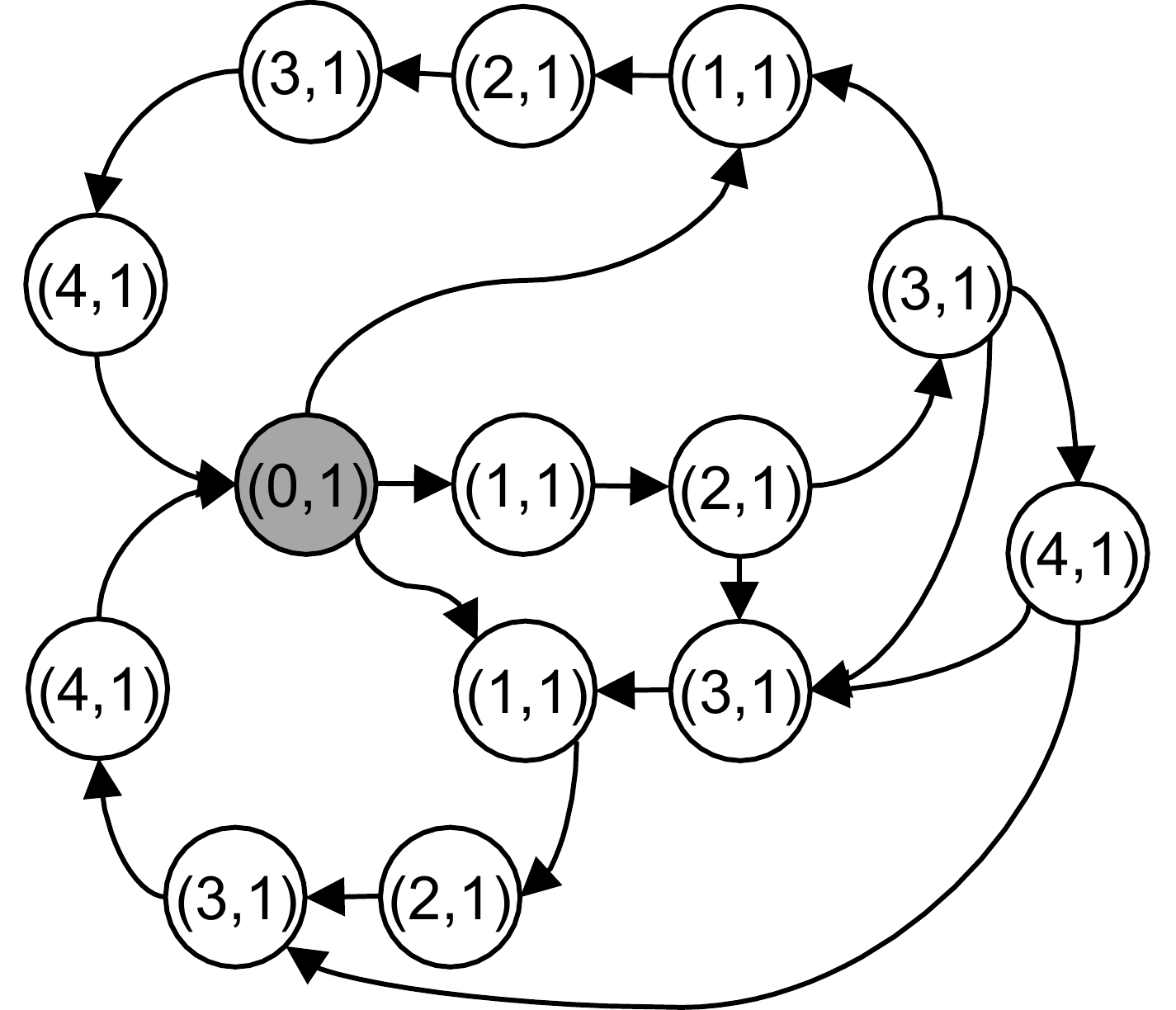}      
		}\hspace{0mm}
		\subfigure[After inserting an edge (in red)]{
			\label{inc_2}
			\centering
			\includegraphics[scale=0.25]{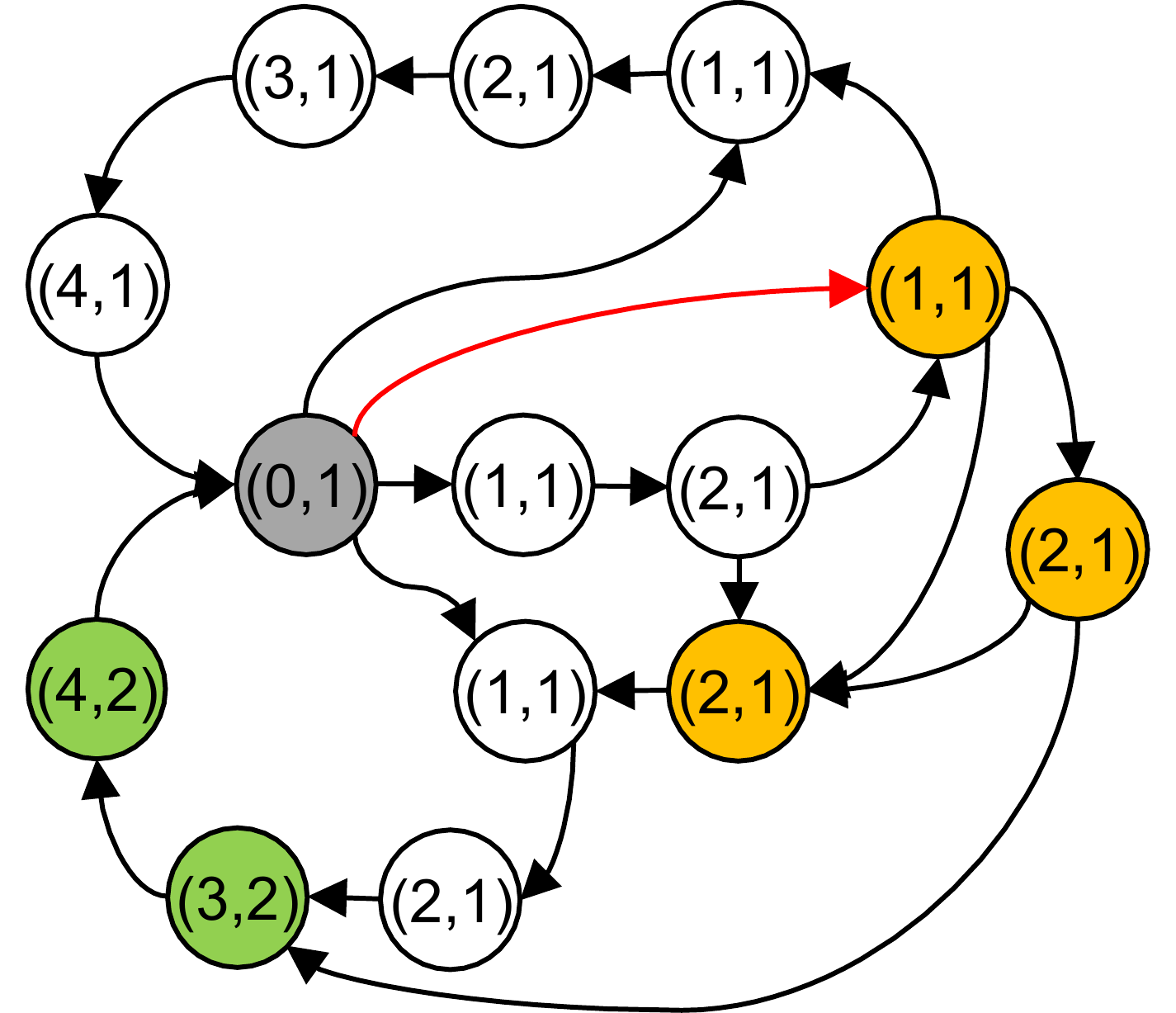}     
		}\hspace{0mm}
	\end{center}
	\vspace{-5mm}
	\caption{An example for the incremental update.} \label{fig:inc_ex}
	\vspace{-3mm}
\end{figure}}

\begin{example}
\autoref{fig:inc_ex} illustrates an example, bipartite conversion is omitted for simplicity. We begin with the grey vertex and assume it has the highest rank in the graph. The first digit of each vertex represents the shortest distance from the grey vertex, while the second digit represents the number of the shortest paths. Additionally, they are the in-label entry for each vertex associated with the grey vertex's hub. When the edge (in red) is inserted in \autoref{fig:inc_ex}(b), yellow vertices alter the shortest distance ($Case_3$), whereas green vertices just affect the shortest path counting ($Case_2$). White vertices are unaffected ($Case_1$ or unvisited).
\end{example}

In addition to the distance pruning, the BFS also terminates if $w\prec v_{k}$, indicating that $v_{k}$ is an eligible hub for $L_{in}(w)$, and vice versa to flip the direction of updating out-labels. Since the query evaluation is to determine the shortest distance via common hubs, out-of-date label entries that do not reflect the shortest distance will be dominated by new label entries after the update. As a consequence, their presence has no bearing on the correctness. The selection of $v_{k}$ and $u_{k}$, that is the affected hubs, is as follows.

\begin{definition}[Affected Hubs]
With the addition of a new edge $e(a, b)$, we define the affected hubs associated with $a$ as $hub_{A} = \{h_{A} | (h_{A},d,c) \in L_{in}(a)\}$ and the affected hubs related to $b$ as $hub_{B} = \{h_{B} | (h_{B},d,c) \in L_{out}(b)\}$.
\label{def:aff_hub}
\end{definition}

Some $(h_{A},d,c) \in L_{in}(\cdot)$ may be out-of-date if $h_{A} \in hub_A$. Some $(h_{B},d,c) \in L_{out}(\cdot)$ may be out-of-date if $h_{B} \in hub_B$. Additionally, some new in-label entries $(h_{A},d,c)$ and new out-label entries $(h_{B},d,c)$ should be generated to restore the covering constraint. Consider $hub_A$, for each $v_{k}$ in $hub_A$, $v_{k}$ is the highest ranked vertex among all or some of the shortest paths from $v_{k}$ to $a$. Thus, they are eligible to pass through $e(a, b)$ and update the index using rank pruning and distance pruning. Vice versa for vertices in $hub_B$ in the reverse direction. Other vertices are trimmed or rendered inaccessible by $a$ in $G_{0}$ (or $b$ in the reverse graph) during the initial index building. The new edge has no effect on the label entries that use these vertices as hubs.

Thus, the search of new shortest paths should be conducted from $hub_A$ in the forward direction and from $hub_B$ in the reverse direction. The shortest path counting information at the beginning of the BFS should be as follows.

\begin{theorem}
If $h \in hub_{A}$ $(or$ $hub_{B})$, and $(h,d,c) \in L_{in}(a)$ $(or$ $L_{out}(b))$. The shortest path counting utilized at the start of the BFS should be $c$ for the affected hub $h$.
\end{theorem}

\begin{proof} 
As shown in \autoref{fig:can_noncan}, $h_{c}$ is a canonical hub in $L_{in}(a)$ and a hub in $L_{out}(h_{n})$, $h_{n}$ is a non-canonical hub in $L_{in}(a)$. According to the constraint of the labeling schema, $h_{c}$ has a higher rank than $h_{n}$. Assume that the current affected hub to be processed is $h_{n}$, we first calculate $sd(h_{n},a)$ and SPCnt$(h_{n},a)$ under the current index. And then use them to start the search of new shortest paths from $b$. For counting labeling, if SPCnt($h_{n},a$) is used, the initial shortest path counting also counts the paths from $h_{n}$ to $a$ via $h_{c}$ if they are also the shortest. The label entries with hub $h_{c}$ are already updated (if hubs are processed in descending order), which indicates the shortest path counting from $h_{n}$ to other vertices via $h_{c}$ is up to date. Thus, the update may overestimate the counting in non-canonical labels if using SPCnt($h_{n},a$) before the search.
\end{proof}

\begin{figure}[t]
  \centering
  \includegraphics[scale=1]{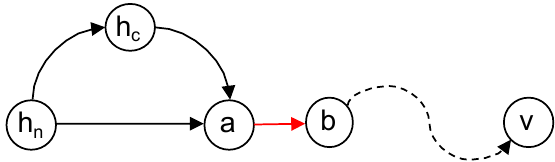}
  \vspace{-2mm}
  \caption{Update for Canonical Hub and Non-Canonical Hub.}
  \vspace{-3mm}
  \label{fig:can_noncan}
\end{figure}

Update algorithm \textsc{Inc}CNT is shown in Algorithm~\ref{alg:inc-cnt}. After identifying affected hubs $V_{k}$, the BFS update procedure will start from each $v_{k}$ in $V_{k}$ in descending order (lines 4-10). The local distance and counting of $v_{k}$ are utilized to start the BFS (lines 6,9).

Algorithm~\ref{alg:forwardpass} illustrates the process of calculating new shortest paths and updating in-labels. When a vertex $w$ is encountered, let $M_{G}$ represent the current index and $D[w]$ denote the tentative distance from $v_{k}$ to $w$, while $D_{G}(v_{k},w)$ specifies the distance calculated using $M_{G}$ (line 7). $C[w]$ keeps track of the tentative shortest path counting. If a new shortest path is discovered ($D[w] \leq D_{G}(v_{k},w)$), update the $L_{in}(w)$ (line 10). The BFS terminates based on the distance and rank pruning (line 12). Algorithm {\textsc{Backward}PASS} is similar to Algorithm~\ref{alg:forwardpass} except that the BFS starts at $a$ in the reverse direction.
We omit it due to the space limits.

Algorithm~\ref{alg:update-label} explains how to update in-labels. If hub $v_{k}$ exists in $L_{in}(w)$ and the new shortest path becomes shorter, replace the label entry with the new distance and counting (line 3). Or if the new shortest path has the same length as before, simply accumulate the counting (line 6). If hub $v_{k}$ does not exist, insert a new label entry (line 8). To update out-labels, replace $L_{in}(\cdot)$
with $L_{out}(\cdot)$. $M_{G_{+}}$ denotes the updated index.


\begin{algorithm}[htb]
\vspace{-1mm}

\setstretch{0.95}
    $hub_{A} \leftarrow$ hubs from $L_{in}(a)$\;
    $hub_{B} \leftarrow$ hubs from $L_{out}(b)$\;
    $V_{k} \leftarrow hub_{A} \cup hub_{B}$\;
    \For{{\bf each} $v_{k}\in V_{k}$ {\rm in descending order}}{
        \If{$v_{k} \in hub_{A}$ {\bf and} $v_{k} \prec b$}{
            $(v_{k},d,c) \leftarrow$ label with hub $v_{k}$ from $L_{in}(a)$\;
            {\textsc{Forward}PASS}($v_{k},b,d+1,c$)\;
        }
        \If{$v_{k} \in hub_{B}$ {\bf and} $v_{k} \prec a$}{
            $(v_{k},d,c) \leftarrow$ label with hub $v_{k}$ from $L_{out}(b)$\;
            {\textsc{Backward}PASS}($v_{k},a,d+1,c$)\;
        }
    }
\caption{\textsc{Inc}CNT($a,b$)}
\vspace{-1mm}
\label{alg:inc-cnt}
\end{algorithm}

\begin{algorithm}[htb]
\setstretch{0.95}
    \For{{\bf each} $v\in V$} {
        $D[v] \leftarrow \infty;$ $C[v] \leftarrow 0$\;
    }
    $D[b] \leftarrow D;$ $C[b] \leftarrow C$\;
    Queue $Q \leftarrow $ $\emptyset$; $Q.{\rm enqueue}(b)$\;
    \While{$Q$ {\rm $\neq \emptyset$}} {
        $w \leftarrow Q.{\rm dequeue}()$\;
        $D_{G}(v_{k},w) \leftarrow $ distance from $v_{k}$ to $w$ under $M_{G}$\;
        \If{$D[w] > D_{G}(v_{k},w)$}{
            {\bf continue};
        }
        \textsc{Update}LABEL(($v_{k},D[w],C[w]$), $L_{in}(w)$)\;
        \For{{\bf each} $u \in nbr_{out}(w)$} {
            \If{$D[u] > D[w] + 1$ {\bf and} $v_{k} \prec u$}{
                $D[u] \leftarrow D[w] + 1$; $C[u] \leftarrow C[w]$\;
                $Q.{\rm enqueue}(u)$\;
            } \ElseIf {$D[u]  = D[w] + 1$} {
                $C[u] \leftarrow C[u] + C[w]$\;
            }
        }
    }
    \For{{\bf each} $v\in V$} {
        $D[v] \leftarrow \infty;$ $C[v] \leftarrow 0$\;
    }
\caption{\textsc{Forward}PASS($v_{k},b,D,C$)}
\label{alg:forwardpass}
\end{algorithm}

\begin{algorithm}[htb]
\setstretch{0.95}
    \If{$(v_{k},d',c') \in L_{in}(w)$}{
        \If{$d < d'$}{
            Replace $(v_{k},d',c')$ with $(v_{k},d,c)$\;
            \textsc{Clean}LABEL($w$, $L_{in}(w)$)\;
        } \ElseIf {$d = d'$} {
            Replace $(v_{k},d',c')$ with $(v_{k},d,c + c')$\;
        }
    } \Else{
        Insert $(v_{k},d,c)$ to $L_{in}(w)$\;
        \textsc{Clean}LABEL($w$, $L_{in}(w)$)\;
    }

\caption{\textsc{Update}LABEL($(v_{k},d,c)$, $L_{in}(w)$)}
\label{alg:update-label}
\end{algorithm}

\begin{theorem}
\vspace{-1mm}
Let $M_{G}$ be the shortest cycle counting index for graph $G_{0}$ and $M_{G_{+}}$ be the index updated by Algorithm~\ref{alg:inc-cnt} from $M_{G}$ regarding the edge insertion to make $G_{+}$ from $G_{0}$. Then, the index $M_{G_{+}}$ is a correct shortest cycle counting index for $G_{+}$. 
\end{theorem}

\begin{proof}
As shown in \autoref{fig:correct_minimal}, if edge $e(a,b)$ generates new shortest paths from $v$ to $w$, there is at least a path from $v$ to $a$ and a path from $b$ to $w$. Let $h_{v}$ and $h_{w}$ denote the hubs of the shortest paths from $v$ to $a$ and the shortest paths from $b$ to $w$, respectively. If $h_{v} \prec h_{w}$, $h_{v}$ will traverse to $w$ in the forward direction and try to update $w$'s in-label. And vice versa for $h_{w}$ to $v$ in the reverse direction if $h_{w} \prec h_{v}$. Hence, after executing \textsc{Inc}CNT(a,b), the higher-ranked one between $h_{v}$ and $h_{w}$ should be the hub of the corresponding shortest paths from $v$ to $w$. As all the label entries with hubs like $h_{v}$ and $h_{w}$ will be updated. Then, for each pair of vertices, their shortest paths through $e(a,b)$ are updated into the index. Label entries related to other shortest paths are unchanged. Therefore, SPCnt($v,w$) is correct under $G_{+}$, and the correctness of all pair shortest path counting guarantees the correctness of SCCnt($u$) for any vertex $u$. 
\end{proof}

\begin{figure}[t]
	\centering
	\includegraphics[scale=1]{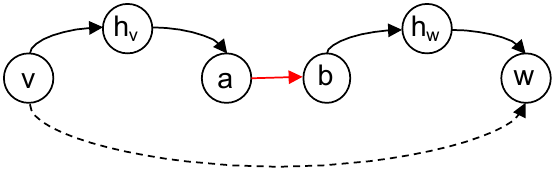}
		\vspace{-3mm}
	\caption{Graph for Proof of Redundancy and Correctness.}
	\vspace{-3mm}
	\label{fig:correct_minimal}
\end{figure}

Moreover, to preserve the minimality property, we need to eliminate those superfluous label entries. Redundant label entries are defined as follows:

\begin{definition}[Redundant Labels]
When a label entry is updated with a shorter distance or a new label entry is inserted, redundancy of the label may occur. A label entry $(h,d,c)$ in $L_{in}(v)$ is redundant if $d > sd(h,v)$ in $G_{+}$; A label entry $(h,d,c)$ in $L_{out}(v)$ is redundant if $d > sd(v,h)$ in $G_{+}$.
\end{definition}

To ensure the minimality, of the label, a redundant label entry check is needed. We examine redundancy situations using the forward direction update as an example, and the cleaning process is shown in Algorithm~\ref{alg:clean-label}. As an example, in \autoref{fig:correct_minimal}, suppose $h_{v}$ is pushed for update and encounters $w$. In this instance, the duplicate label entries correspond to the paths towards $w$ (dash line in \autoref{fig:correct_minimal}). There are two situations of the start vertex $v$ of such paths:

\begin{itemize}
  \item If $v \prec w$, then $(v,d,c)$ is a label entry in $L_{in}(w)$. Following index updates, $sd(v,w)$ must be re-evaluated. If $d > sd(v,w)$, then all the shortest paths from $v$ to $w$ pass via $e(a,b)$, resulting in the redundancy of $(v,d,c)$ (lines 1-5);
  \item If $w \prec v$, then $(w,d,c)$ is a label entry in $L_{out}(v)$. Similarly, we determine if $sd(v,w)$ is less than $d$ when the index is updated. To facilitate implementation, an out-label inverted index $inv\_out(w)$ is added in this instance to locate the vertices similar to $v$ whose out-label contains the hub $w$ (lines 6-11). An inverted index of this kind can be constructed during the initial index creation process.
\end{itemize}

For the reverse direction cleaning, the process is similar. For the first instance, We examine the out-label and utilize an in-label inverted index ($inv\_in(\cdot)$) to locate and delete redundant in-label entries.

\begin{algorithm}[t]
\setstretch{0.95}
    \For{{\bf each} $(h,d,c) \in L_{in}(w)$}{
        $D_{G_{+}}(h,w)\leftarrow$ distance from $h$ to $w$ under $M_{G_{+}}$\;
        \If{$d > D_{G_{+}}(h,w)$}{
            remove $(h,d,c)$ from $L_{in}(w)$\;
            remove $w$ from $inv\_in(h)$\;
        }
    }
    \For{{\bf each} $v \in inv\_out(w)$}{
        $(w,d,c)\leftarrow$ the label from $L_{out}(v)$\;
        $D_{G_{+}}(v,w)\leftarrow$ distance from $v$ to $w$ under $M_{G_{+}}$\;
        \If{$d > D_{G_{+}}(v,w)$}{
            remove $(w,d,c)$ from $L_{out}(v)$\;
            remove $v$ from $inv\_out(w)$\;
        }
    }
\caption{\textsc{Clean}LABEL($w$, $L_{in}(w)$)}
\label{alg:clean-label}
\end{algorithm}

\begin{theorem}
If Algorithm \ref{alg:clean-label} is used to update the shortest path counting index, the index is \textit{minimal}. As a consequence, the lack of any label entries leads to erroneous counting query results for certain vertices.
\end{theorem}

\begin{proof}
Algorithm \ref{alg:clean-label} deletes all superfluous label entries. This indicates that for any $(h_{v},d_{v},c_{v})\in L_{out}(v)$, $d_{v}=sd(v,h_{v})$, $h_{v}$ has the highest rank along some shortest paths from $v$ to $h_{v}$, and the number of which is $c_{v}$. For any $(h_{w},d_{w},c_{w})\in L_{in}(w)$, $d_{w}=sd(h_{w},w)$, there exist $c_{w}$ shortest paths from $h_{w}$ to $w$ with $h_{w}$ as the highest rank vertex. We then demonstrate that the removing of any $(h_{v},d_{v},c_{v})\in L_{out}(v)$ or $(h_{w},d_{w},c_{w})\in L_{in}(w)$ will results in an erroneous SPCnt($v,w$). Assume that $h$ is a common hub in $L_{out}(v)$ and $L_{in}(w)$ and that it generates a portion of the shortest paths $P_{h}$ from $v$ to $w$. Then, among all vertices along $P_{h}$, $h$ has the highest rank. Assume that $(h,d_{h},c_{h})$ has been removed from $L_{out}(v)$ (or $L_{in}(v)$). Another hub $k$ exists that could induce $P_{h}$ where $k\neq h$. Thus, $k$ should be the highest ranked vertex along $P_{h}$. This implies that $k=h$ and therefore refutes the assumption. Elimination of any $(h_{v},d_{v},c_{v})\in L_{in}(v)$ will at the very least result in $SPCnt(h_{v},v)$ is the incorrect answer. And vice versa for any out-label entry.
\end{proof}

\subsection{Analysis}
\subsubsection*{Efficiency Trade-off}~The time required to remove redundant labels associated with $v$ is dependent on the sum of $|L_{in}(v)|$ and $|inv\_out(v)|$ (or $|L_{out(v)}|$ and $|inv\_in(v)|$). Cleaning takes $O(k\omega\log n)$ time, where $k$ is the above-mentioned size and is often considerably longer than updating with redundancy, and $\omega$ is the treewidth from \autoref{the:complexity}. In practice, we choose to update with redundancy in order to maximize efficiency (skip lines 4 and 9 in Algorithm~\ref{alg:update-label}).

\subsubsection*{ Time Complexity}~We begin with $O(2\omega \log n)$ affected vertices. Assume that during the resumed BFS, $k$ vertices are visited and a query is processed for each of them. Thus, the time complexity of adding a new edge is $O(k\omega^{2}\log^{2}n)$.

\subsection{Edge Deletion}
The deletion of an edge may result in an increase in the distance between specific pairs of vertices or a reduction in the shortest path counting between them. To ensure consistency, certain label entries must be deleted. Consider the case where $(v,d,c)\in L_{in}(w)$, but $v$ is detached from $w$ after the edge deletion. In this instance, $(v,d,c)$ results in an erroneous answer of SPCnt$(v,w)$ and SCCnt$(u)$ for specific $u$ when there are shortest cycles through $u$ via $v$ and $w$. 

The decremental update method is comparable to the incremental technique, which could be divided into three steps. Due to space limits, we will only briefly explain them. Assume the omitted edge is $e(a,b)$. We begin by identifying two sets of impacted hubs $hub_{A}$ and $hub_{B}$. Each vertex $v$ in $hub_{A}$ must fulfill the constraint $sd(v,a) + 1 = sd(v,b)$. Each vertex $u$ in $hub_{B}$ must fulfill the constraint $sd(b,u) + 1 = sd(a,u)$. Note that $hub_{A}$ and $hub_{B}$ may have identical vertices that are located on a cycle through $(a,b)$. Delete out-of-date label entries in the second step as follows: If $v\in hub_{A} \wedge u\in hub_{B}$, delete $(v,d,c)$ from $L_{in}(u)$ or delete $(u,d,c)$ from $L_{out}(v)$ if they exist. The set of deleted label entries is a superset of out-of-date ones. The distances from or towards the vertices which are other than $hub_{A}\cup hub_{B}$ are unaffected by $e(a,b)$. The last step is to add label entries. BFS is carried out starting at each vertex $v$ in $hub_{A}$ and inserting $(v,d,c)$ into $L_{in}(u)$ if $u\in hub_{B}$. Vice versa to add out-label entries from vertices in $hub_{B}$ in reverse direction.

\section{Evaluation}
\label{sect:exp}
\subsection{Experiments Setup}
\subsubsection*{Settings} Experiments are conducted on a Linux server with Intel Xeon E3-1220 CPU and 520GB memory. The algorithms are implemented in C++ and compiled  by g++ at -O3 optimization level. Each label entry is encoded in a 64-bit integer. The vertex ID, distance, and counting take 23, 17, and 24 bits, respectively.

\subsubsection*{Datasets} Nine networks from SNAP\footnote{https://snap.stanford.edu\label{snap}} and Konect\footnote{https://konect.cc\label{konect}} were used for the experiments. The details are shown in \autoref{tab:graphs}.
All graphs are directed and have no self-loop. For query evaluation, all vertices of each graph, or at least 50,000 vertices were used, and they were divided into five clusters according to their min-in-out degrees $\min(|nbr_{in}(v)|,|nbr_{out}(v)|)$. We first obtained the highest and lowest degree within each graph. Then divided the degree range evenly into five clusters, High, Mid-high. Mid-low, Low, and Bottom. Vertices were finally clustered based on their min-in-out degrees. For dynamic maintenance, [200,500] random edges were removed and then inserted back to each graph.

\begin{table}[htb]
\centering
\vspace{-5mm}
\caption{The Statistics of The Graphs}
\vspace{-2mm}
  \begin{tabular}{|l|l|l|l|}
    \hline
    \cellcolor{gray!25}\textbf{Graph} & \cellcolor{gray!25}\textbf{Notation} & \cellcolor{gray!25}\textbf{n} & \cellcolor{gray!25}\textbf{m}\\
    \hline
     p2p-Gnutella04\textsuperscript{\ref {snap}}  & G04 & 10,879 & 39,994  \\
    \hline 
     p2p-Gnutella30\textsuperscript{\ref {snap}} & G30 & 36,682 & 88,328\\
    \hline
     email-EuAll\textsuperscript{\ref {snap}} & EME & 265,214 & 420,045 \\
    \hline
     web-NotreDame\textsuperscript{\ref {snap}} & WBN & 325,729 & 1,497,134\\
    \hline
     wiki-Talk\textsuperscript{\ref {snap}} & WKT & 2,394,385 & 5,021,410\\
    \hline
     web-BerkStan\textsuperscript{\ref {snap}} & WBB & 685,231 & 7,600,595\\
    \hline
    Hudong-Related\textsuperscript{\ref {konect}} & HDR & 2,452,715 & 18,854,882\\
    \hline
    wiki$\_$link$\_$War\textsuperscript{\ref {konect}} & WAR & 2,093,450 & 38,631,915\\
    \hline
    wiki$\_$link$\_$SR\textsuperscript{\ref {konect}} & WSR & 3,175,009 & 139,586,199\\
  \hline
\end{tabular}
\label{tab:graphs}
\vspace{-2mm}
\end{table}

\subsubsection*{Compared Algorithms} We compared the following algorithms (i) HP-SPC (Baseline) and (ii) CSC (Proposed algorithm) for the static shortest cycle counting experiments based on the index construction time, index size, and query time for each graph. In addition, we used the na\"ive BFS for the comparison of the query time. 

\subsection{Experiment Results on Static Graph}

\subsubsection{Index Time} \autoref{fig:ind_t_s}(a) shows the index construction time taken by HP-SPC and CSC. The results indicate that 
(i) for graphs EME, WBN, and WKT, even though HP-SPC is 1.22 to 1.38 times faster than CSC in constructing index, the time differences are not significant. For instance, the index time ratio between CSC and HP-SPC for graph EME is 1.38. The time difference is 4.5s; 
(ii) the index construction times of CSC for other tested graphs are slightly longer than HP-SPC. The time difference does not exceed 8\%. When the number of edges exceeds ten million, the time difference is reduced to within 1.5\%; 
(iii) CSC can index any tested graphs with less than ten million edges in 24 minutes. For two graphs with tens of millions of edges, index time of HDR is 15.5 hours, while for WAR, it only requires 0.5 hour. It takes the longest to index the largest graph WSR, roughly 61 hours.

\subsubsection{Index Size} \autoref{fig:ind_t_s}(b) illustrates the results in index size. It is expected that the two algorithms generate a similar size index, and the actual result proves it. We observe that the index size of CSC is nearly the same as that of HP-SPC. The two most significant differences in index size between the two algorithms are 4.4\% and 2.7\% for graph WBN and EME, respectively. At the same time, the index differences for other graphs are all less than 1\%.

\begin{figure}[htb]
\vspace{-4mm}
	\begin{center}
		\subfigure[Index Time]{
			\label{fig:ind_t}
			\centering
			\includegraphics[scale=0.17]{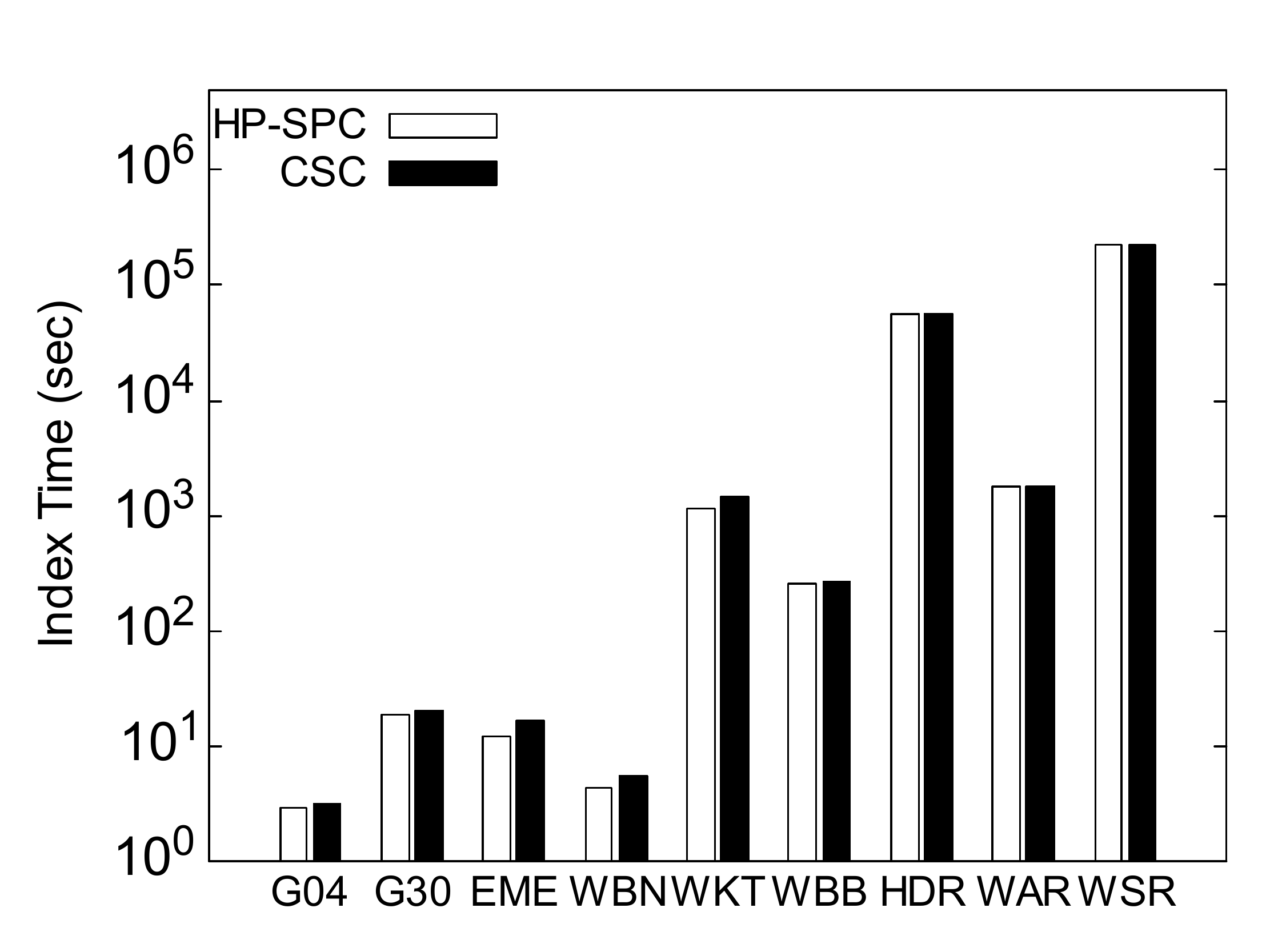}      
		}
		\subfigure[Index Size]{
			\label{fig:ind_s}
			\centering
			\includegraphics[scale=0.17]{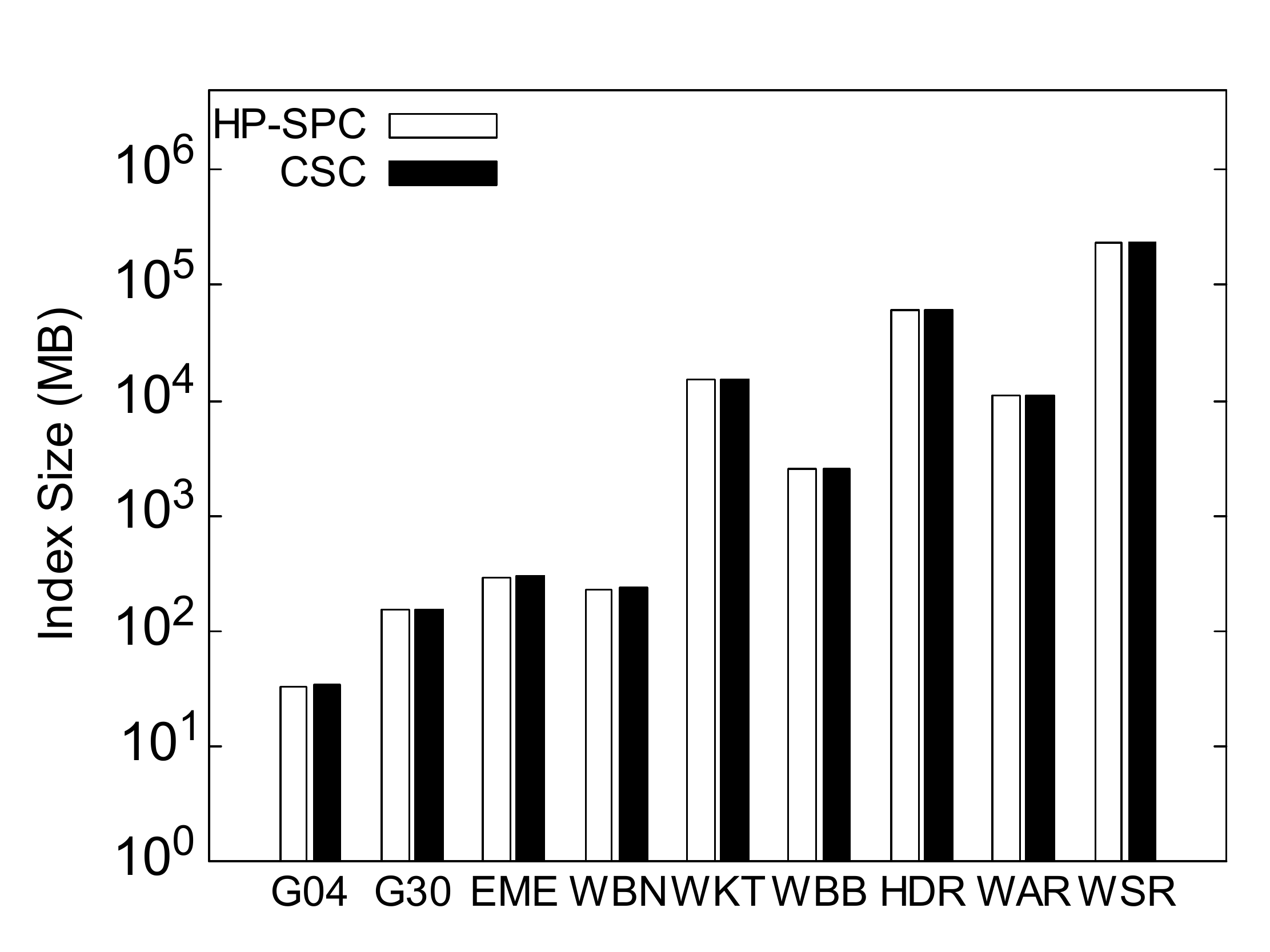}     
		}
	\end{center}
	\vspace{-4mm}
	\caption{Index Time (sec) and Index Size (MB).} \label{fig:ind_t_s}
	\vspace{-3mm}

\end{figure}

\begin{figure}[htb]
\vspace{-2mm}
	\begin{center}
		\subfigure[G04]{
			\label{q_g04}
			\centering
			\includegraphics[scale=0.13]{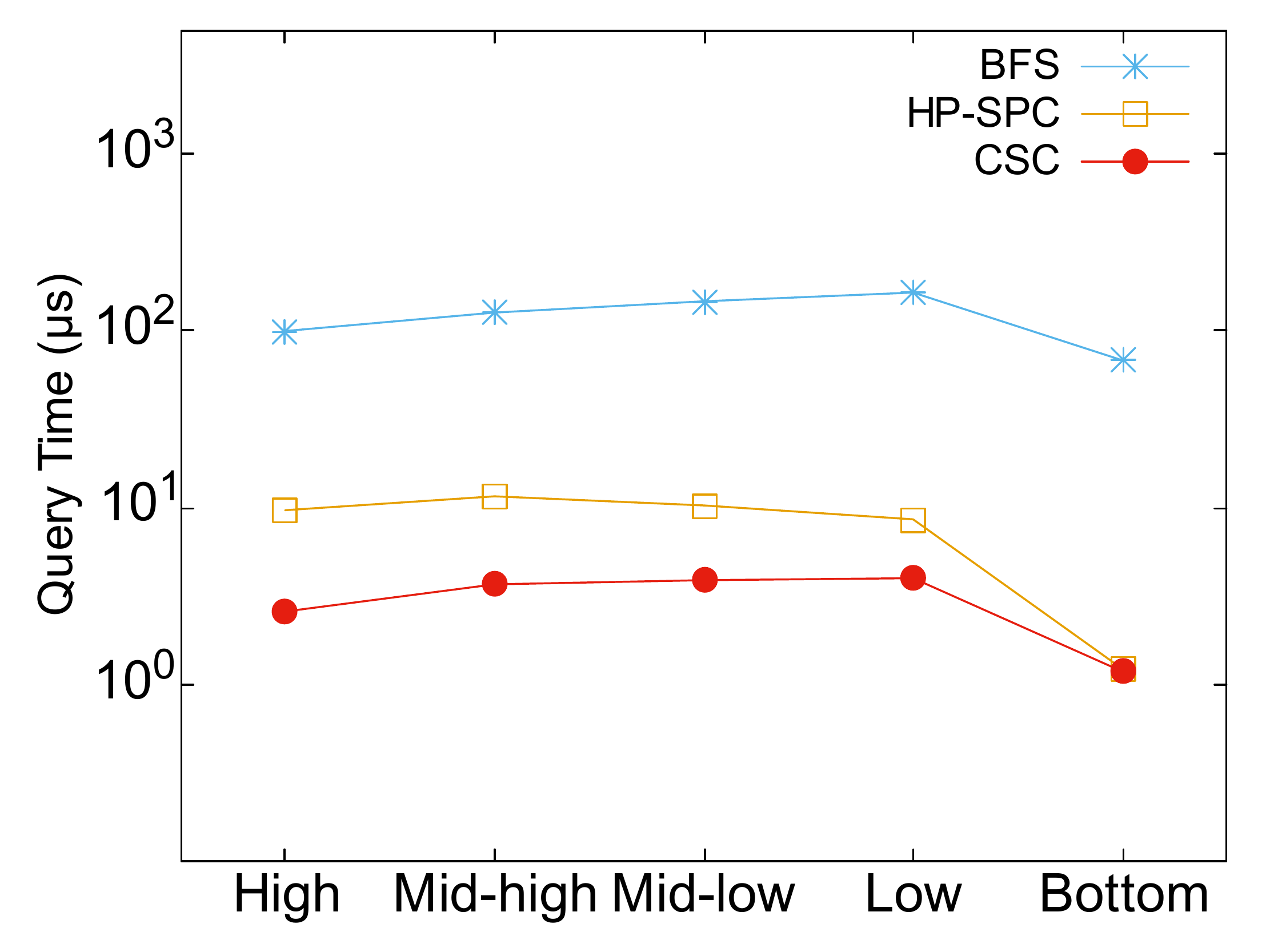}      
		}\hspace{-5mm}
		\subfigure[G30]{
			\label{q_g30}
			\centering
			\includegraphics[scale=0.13]{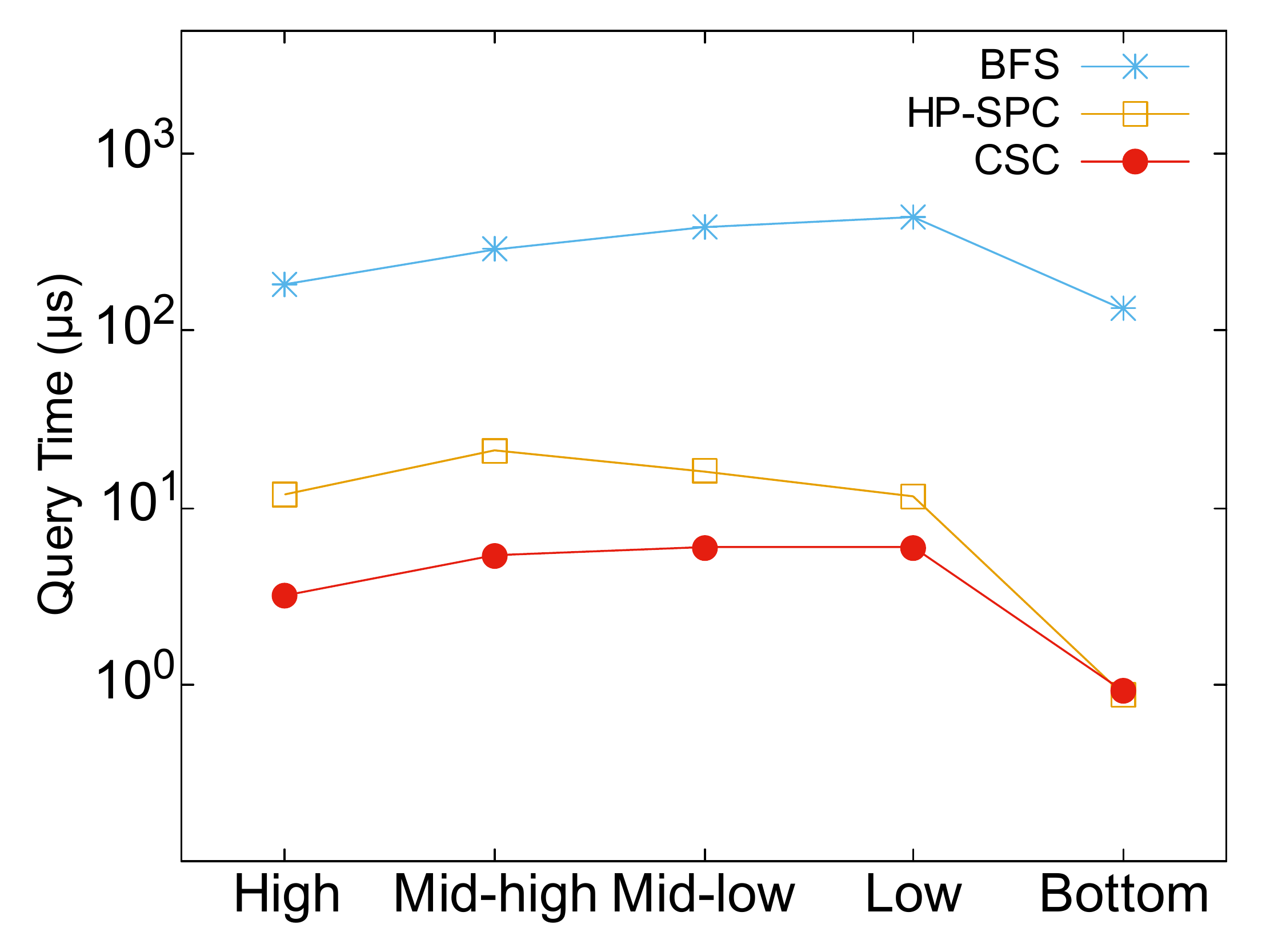}     
		}\hspace{-5mm}
		\subfigure[EME]{
			\label{q_eme}
			\centering
			\includegraphics[scale=0.13]{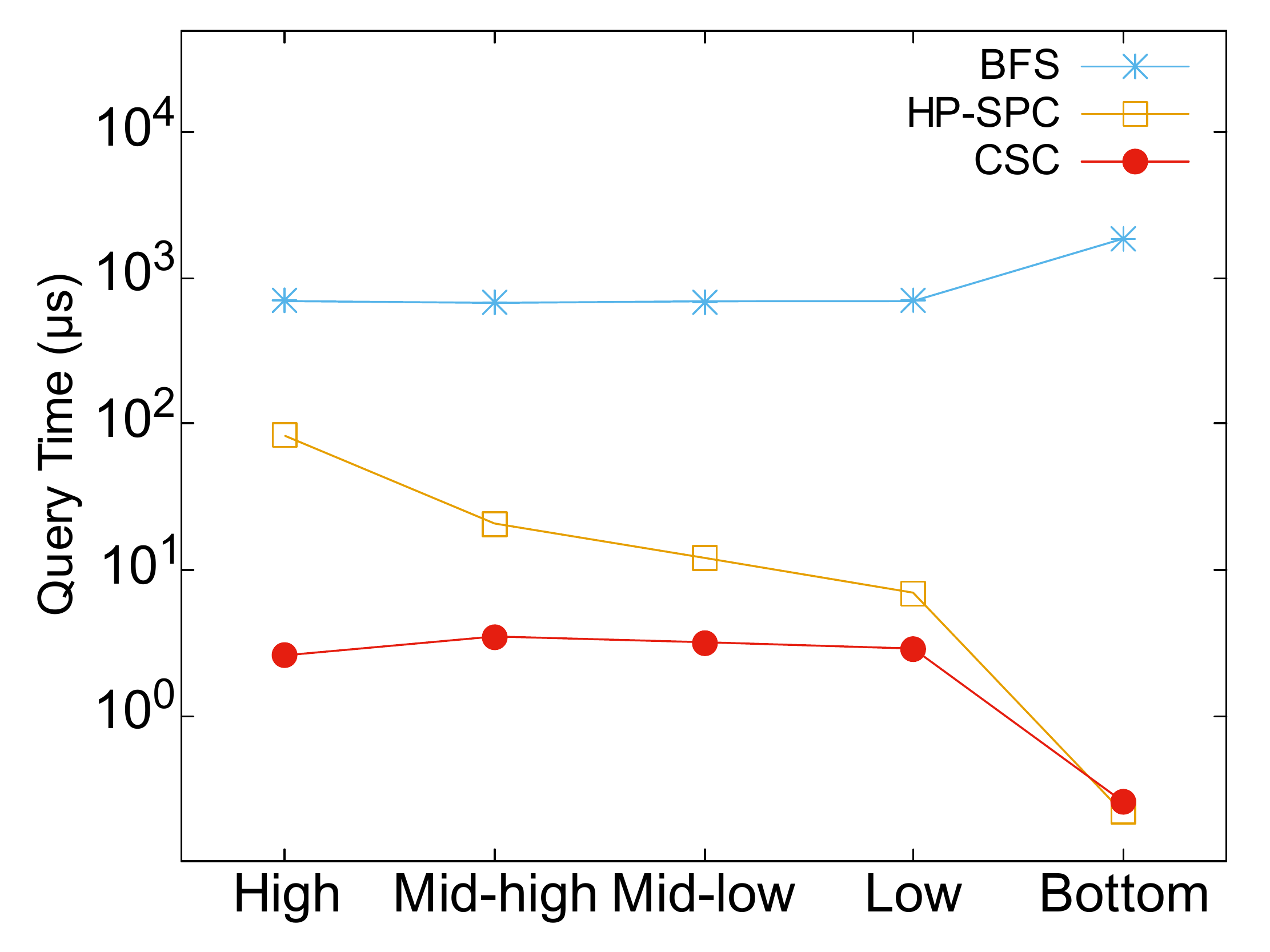}    
		}\hspace{-5mm}
		\subfigure[WBN]{
			\label{q_wbn}
			\centering
			\includegraphics[scale=0.13]{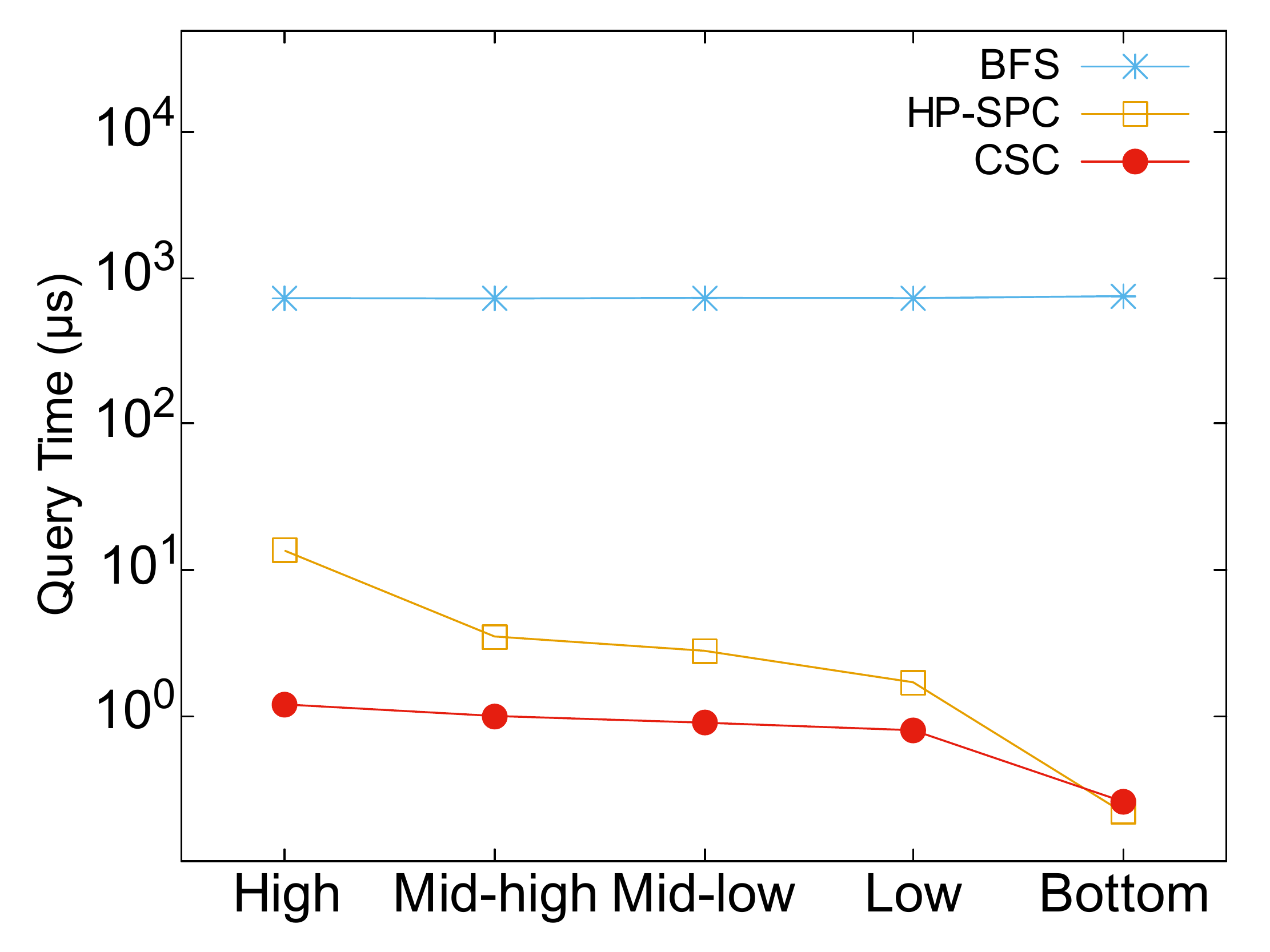}      
		}\hspace{-5mm}
		\subfigure[WKT]{
			\label{q_wkt}
			\centering
			\includegraphics[scale=0.13]{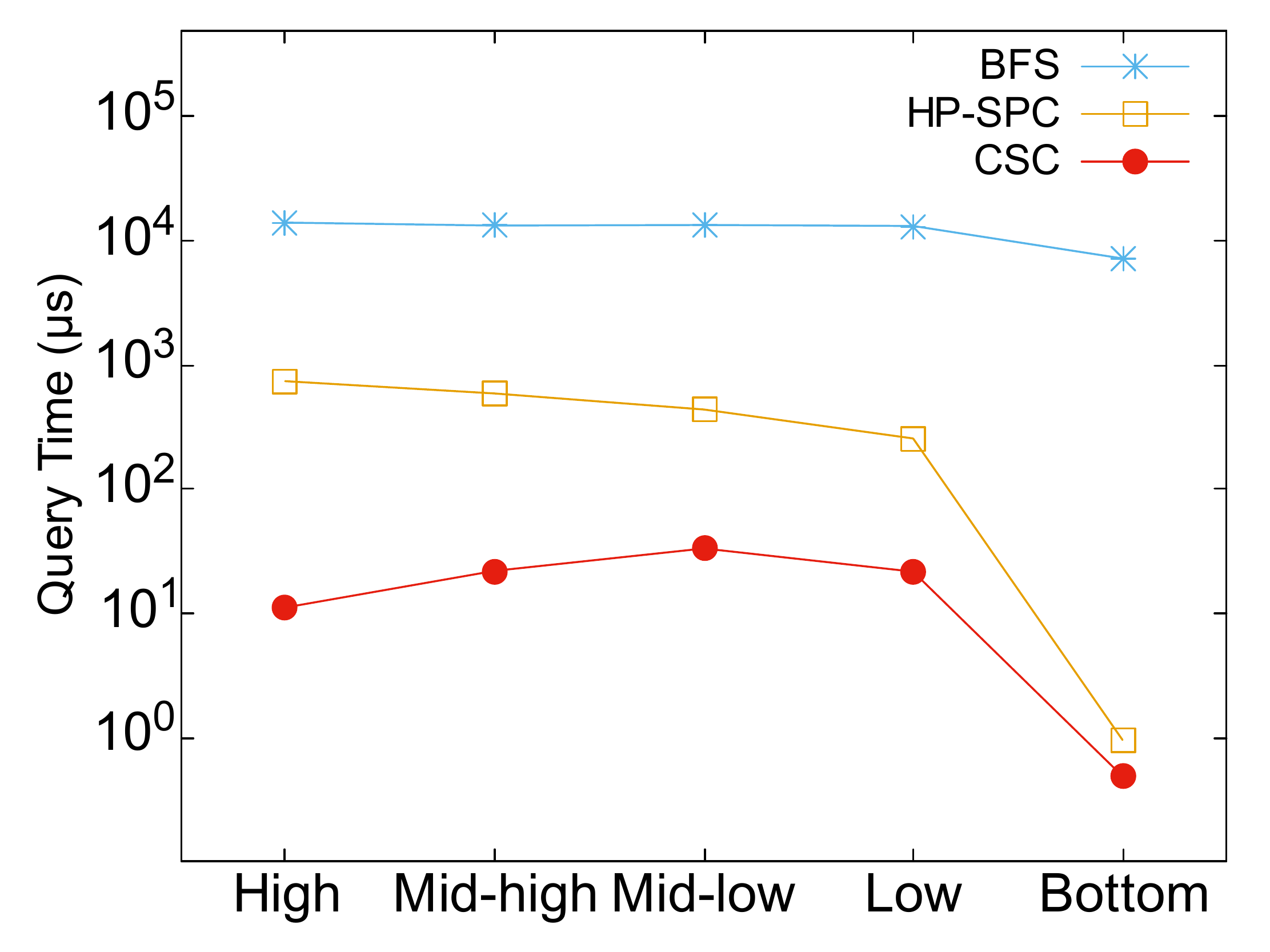}      
		}\hspace{-5mm}
		\subfigure[WBB]{
			\label{q_wbb}
			\centering
			\includegraphics[scale=0.13]{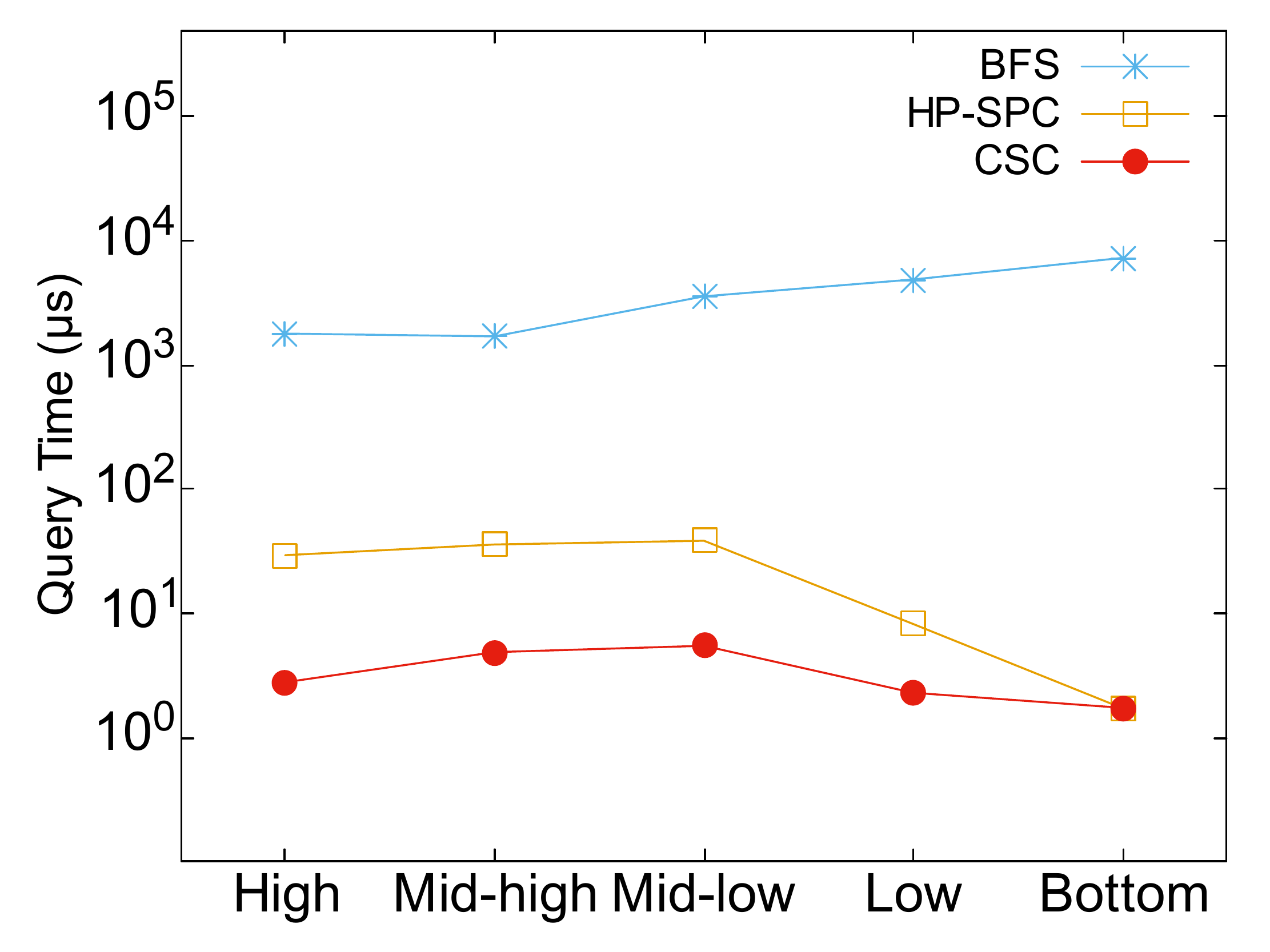}      
		}\hspace{-5mm}
		\subfigure[HDR]{
			\label{q_hdr}
			\centering
			\includegraphics[scale=0.13]{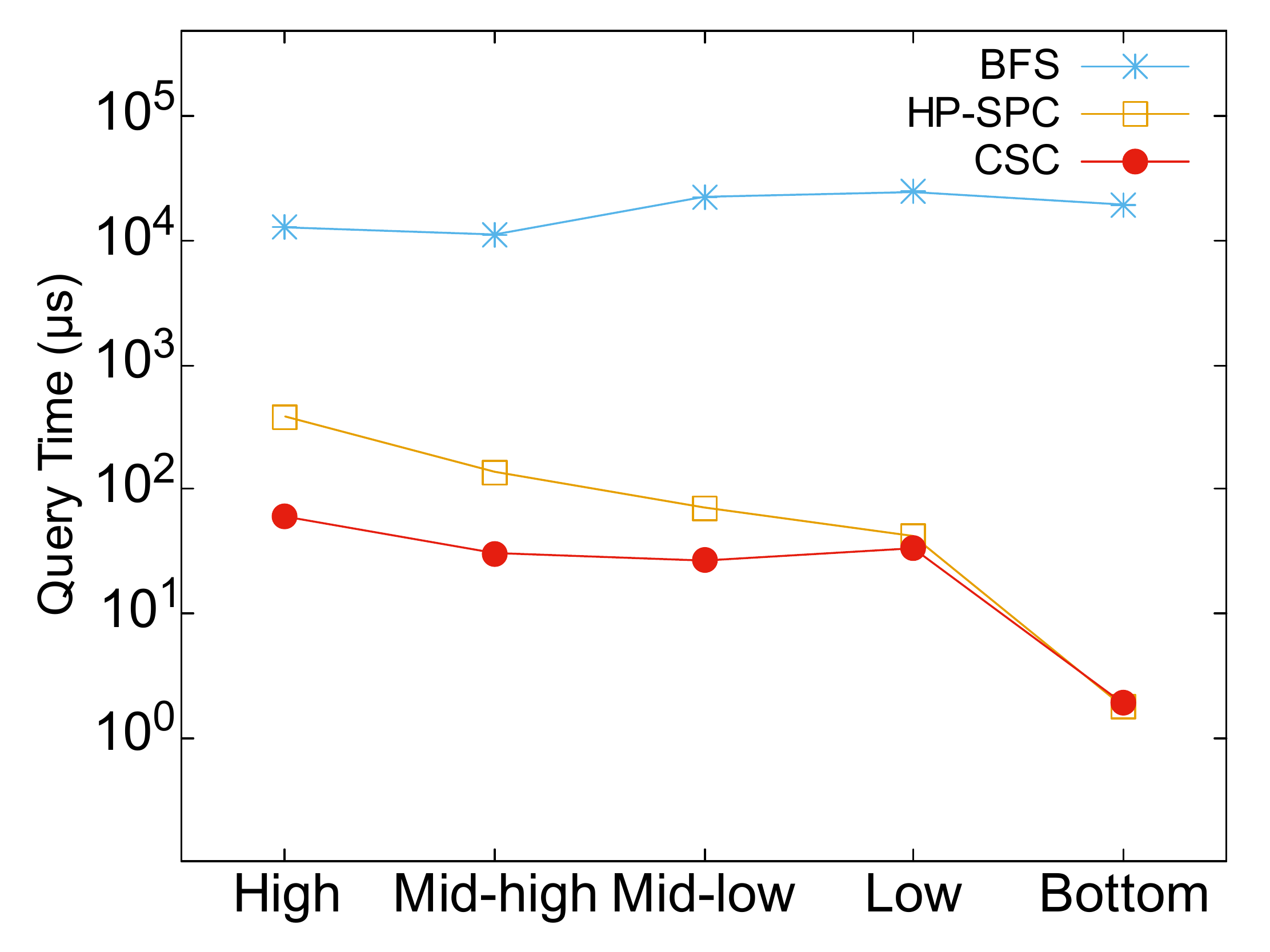}      
		}\hspace{-5mm}
		\subfigure[WAR]{
			\label{q_war}
			\centering
			\includegraphics[scale=0.13]{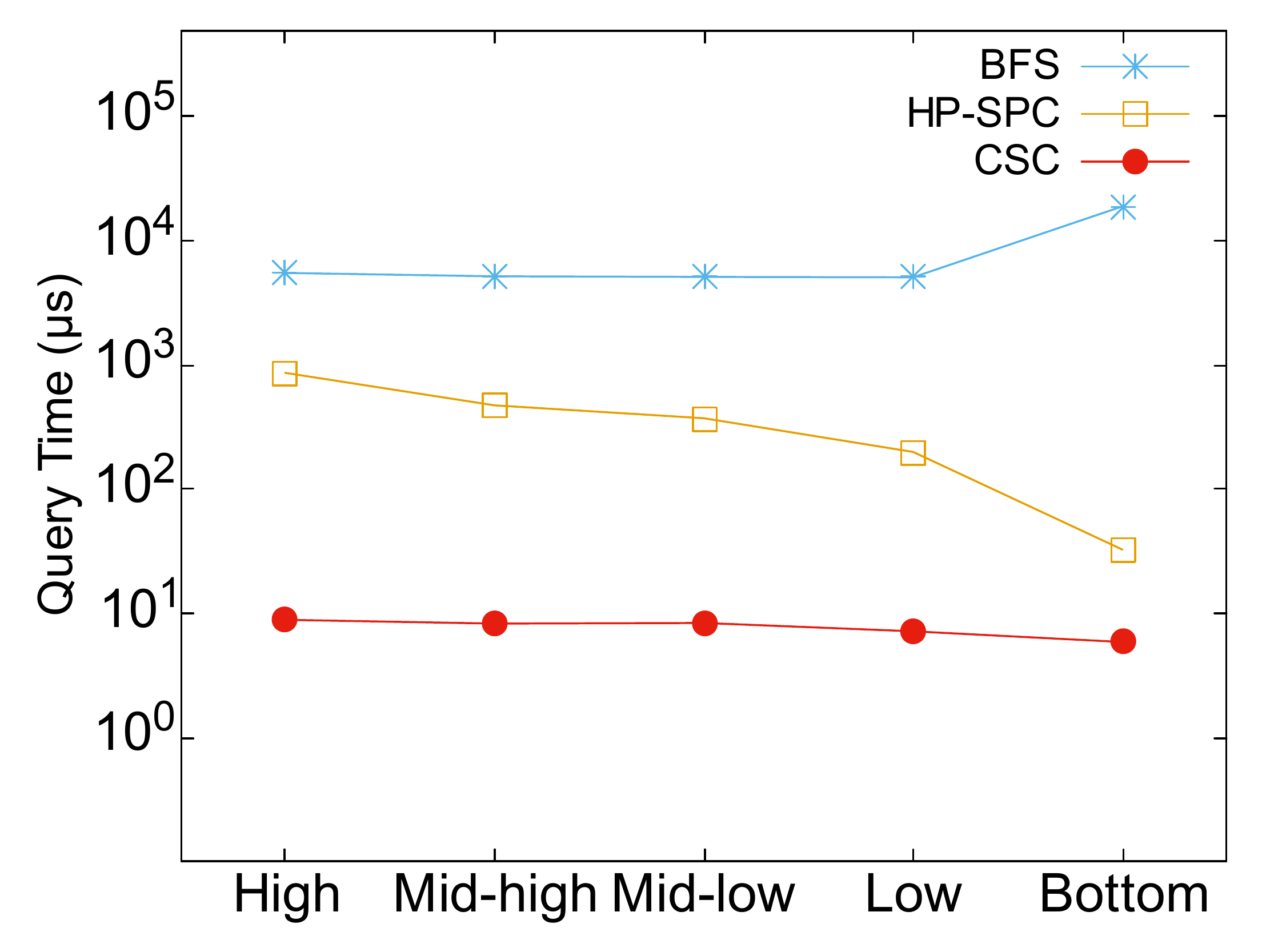}      
		}\hspace{-5mm}
		\subfigure[WSR]{
			\label{q_wsr}
			\centering
			\includegraphics[scale=0.13]{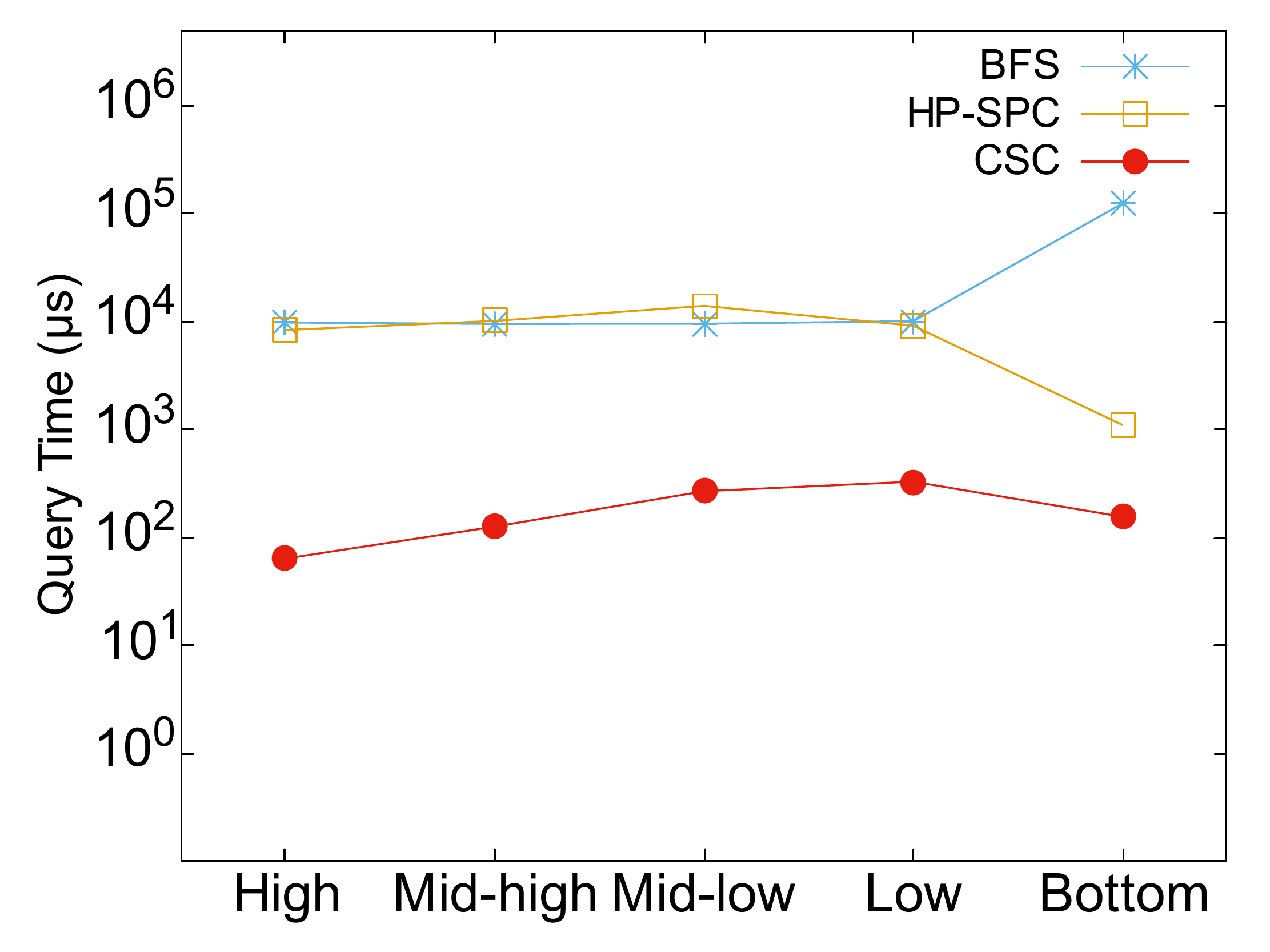}      
		}
	\end{center}
	\vspace{-5mm}
	\caption{Query Times ($\mu s$).} \label{fig:qer_t}
	\vspace{-3mm}
\end{figure}

\subsubsection{Query Time} \autoref{fig:qer_t} illustrates the average query time of each cluster for each graph taken by BFS, HP-SPC, and CSC. There is no evidence that vertex degree has an influence on query time under a na\"ive BFS approach, while its query time is always at a very high value. As shown in each sub-figure, the query time of CSC is more stable than that of HP-SPC and can stay at microsecond level all the time even for the largest graph.
There was not a significant positive linear correlation between the query time of HP-SPC and the degree of query vertices. Because the HP-SPC method involves searching through $v$'$s$ neighbors to answer an SCCnt($v$) query, the query time is $deg_m*t_{P}$. Notably, $deg_m=\min(|nbr_{in}(v)|, |nbr_{out}(v)|)$ and $t_{P}$ is the query time of a single SPCnt evaluation. A higher-degree query vertex has a larger $deg_m$. However, the variety of $t_{P}$ cannot be identified in the interim as we do not know in advance the label sizes of the query vertices, which may be very varied. As a result, the query time does not increase linearly with the degree of query vertex all the time. Nevertheless, evaluating a single SPCnt query is very efficient. Thus, $deg_m$ is a critical component affecting the final query time. If $deg_m$ is large, HP-SPC may take more time to evaluate an SCCnt query. Nevertheless, for those higher-ranked query vertices (in High and Mid-high clusters), HP-SPC takes much more time to evaluate an SCCnt query. Meanwhile, CSC is from 3.11 to 130.1 times faster than HP-SPC among these clusters. Due to the couple-vertex skipping approach, the label size of each vertex is almost the same under CSC as it is under HP-SPC. As a result, the query time of CSC is quite similar to the time required for a single SPCnt evaluation. CSC is no longer required to search through the neighbors. When the query vertex's degree is very large, CSC should have significantly improved query performance. Note that in graphs WKT, WAR and WSR, the average query time for High cluster with CSC is up to two orders of magnitude faster than HP-SPC. For those lower-ranked query vertices (in Mid-low and Low clusters), CSC is still
1.3 to 51.8
times faster. For those vertices with a very small degree, the two algorithms provide similar performance. In summary, these results indicate that CSC is efficient and stable for SCCnt queries with microseconds responses.

\subsection{Experiment Results on Dynamic Graph}

\begin{figure}[htb]
\vspace{-7mm}
	\begin{center}
		\subfigure[Average Update Time]{
			\label{upd_t}
			\centering
			\includegraphics[scale=0.17]{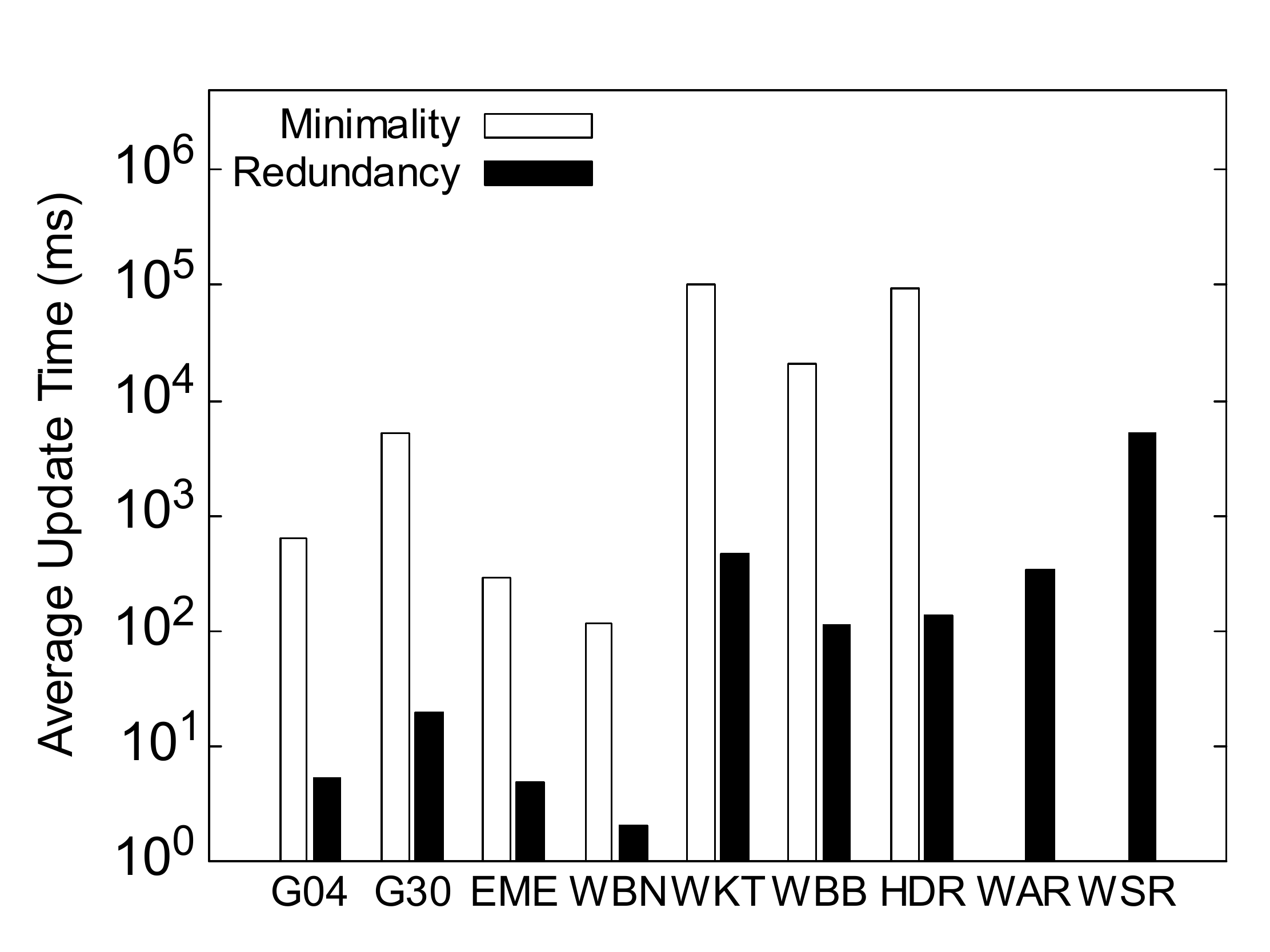}      
		}
		\subfigure[Index Increase Size]{
			\label{upd_s}
			\centering
			\includegraphics[scale=0.17]{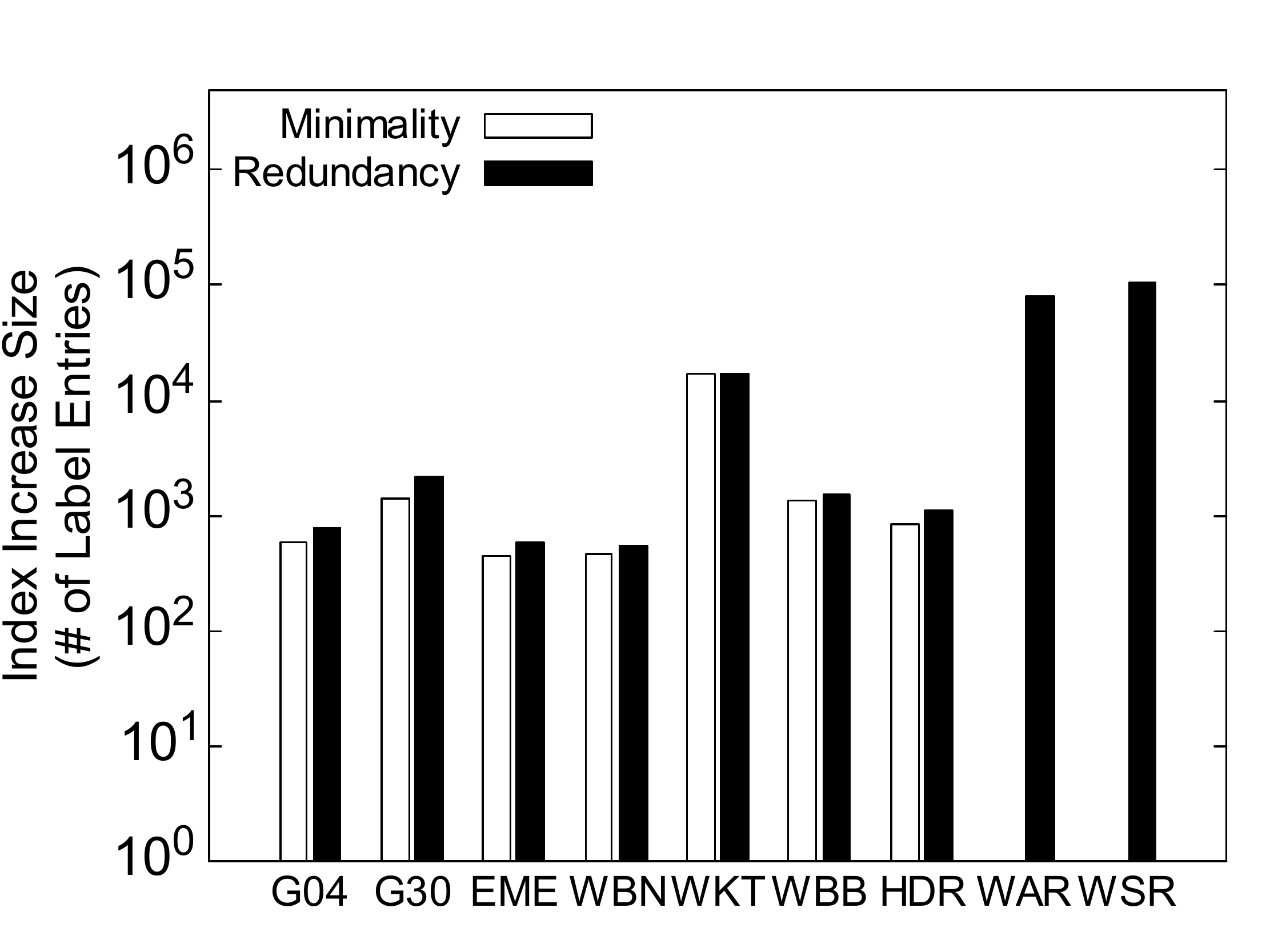}     
		}
	\end{center}
	\vspace{-5mm}
	\caption{Average Update Time (ms) and Index Increase Size (Number of Label Entries) of Incremental Maintenance.} \label{fig:upd_t_s} 
	\vspace{-2mm}
\end{figure}

\subsubsection{Update Time} \autoref{fig:upd_t_s}(a) shows the average incremental update time under both minimality and redundancy strategy. 
The update time was observed to have a similar trend as the index construction time for the graphs with less than ten million edges. The similar process of \textsc{Inc}CNT as that of CSC can explain this. However, \textsc{Inc}CNT is pruned earlier and more often than CSC with different conditions. The update under minimality strategy is 58 to 678 times slower than the redundancy approach. This is because of the large number of redundancy checks. Due to the time cost of minimality strategy, it is omitted for graphs WAR and WSR.
When the graph size becomes larger, the update time doesn't increase too much. It benefits from the small world phenomenon \cite{watts1999networks}. In most cases, the change of an edge will not affect the shortest distance in a large range. Under redundancy strategy, although graphs WKT, WBB, HDR, WAR vary in edge size from millions to tens of millions, their update times are similar, ranging from 112ms to 469ms. For the largest graph WSR, we set the update time limit as 60 seconds (s). If an update cannot be completed within the time limit, it is terminated and the update time is set as 60s. The average update time of WSR under redundancy strategy is 5.2s. Compared with a reconstruction for each edge insertion,  \textsc{Inc}CNT only requires 2.3$\times 10^{-5}$ of the reconstruction time for a single update. The results prove the effectiveness and the scalability of \textsc{Inc}CNT.

\autoref{fig:upd_dec}(a) illustrates the average decremental update time of Graph G04. 
For an edge $e(v,w)$, the degree of this edge is defined as the sum of in-degree of $v$ and out-degree of $w$. These 500 edges are divided into five clusters according to their edge degrees. The High cluster contains those edges with the largest edge degrees. Following the addition of edge degree, a significant increase in the update time was recorded. The average update time of an edge in High cluster is about 2.6s while it is around 0.25s for the edges in Bottom cluster. The total average update time is 0.93s.

\begin{figure}[htb]
	\begin{center}
		\vspace{-5mm}
		\subfigure[Average Update Time]{
			\label{upd_dec_t}
			\centering
			\includegraphics[scale=0.17]{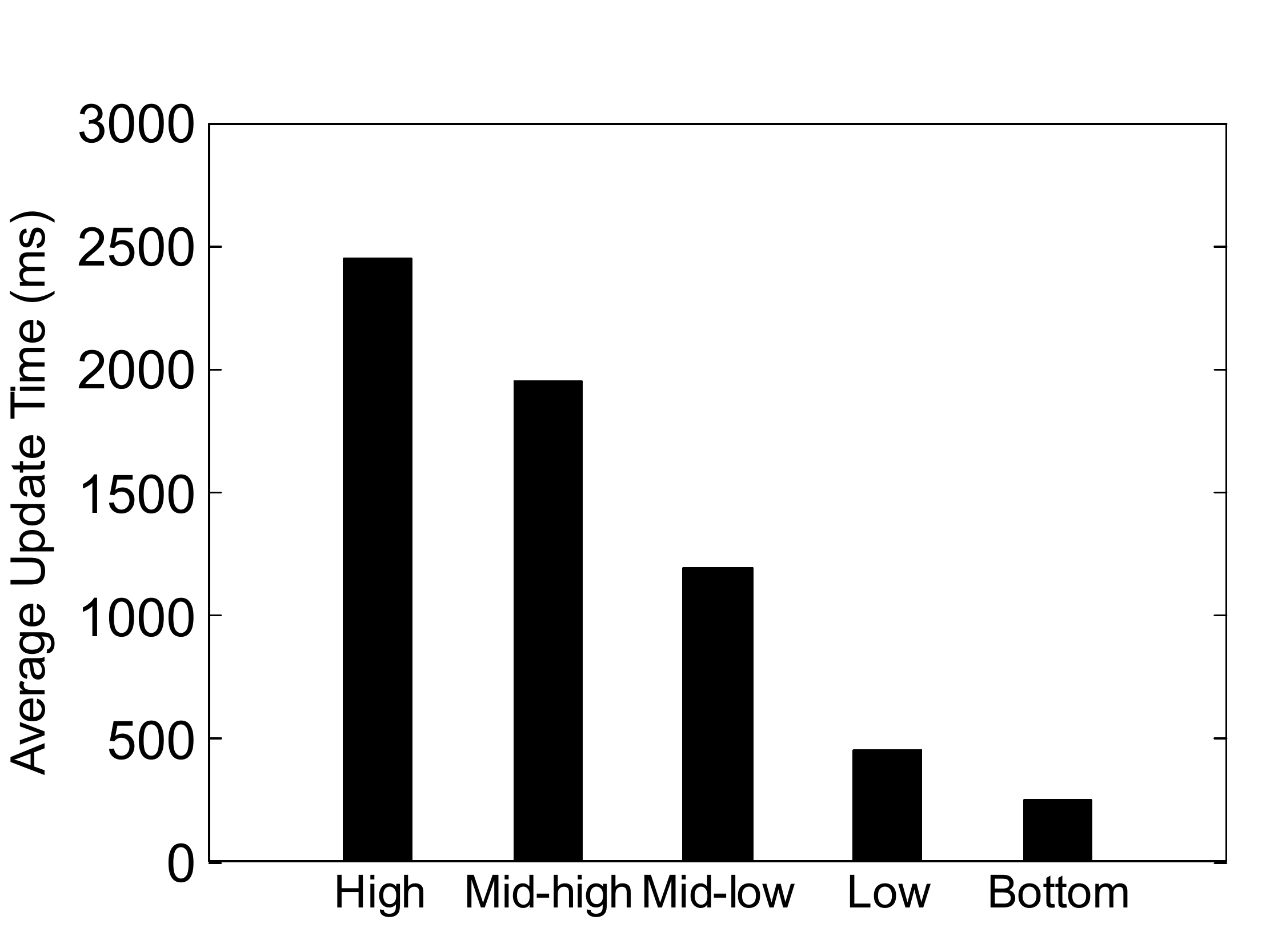}      
		}
		\subfigure[Index Decrease Size]{
			\label{upd_dec_s}
			\centering
			\includegraphics[scale=0.17]{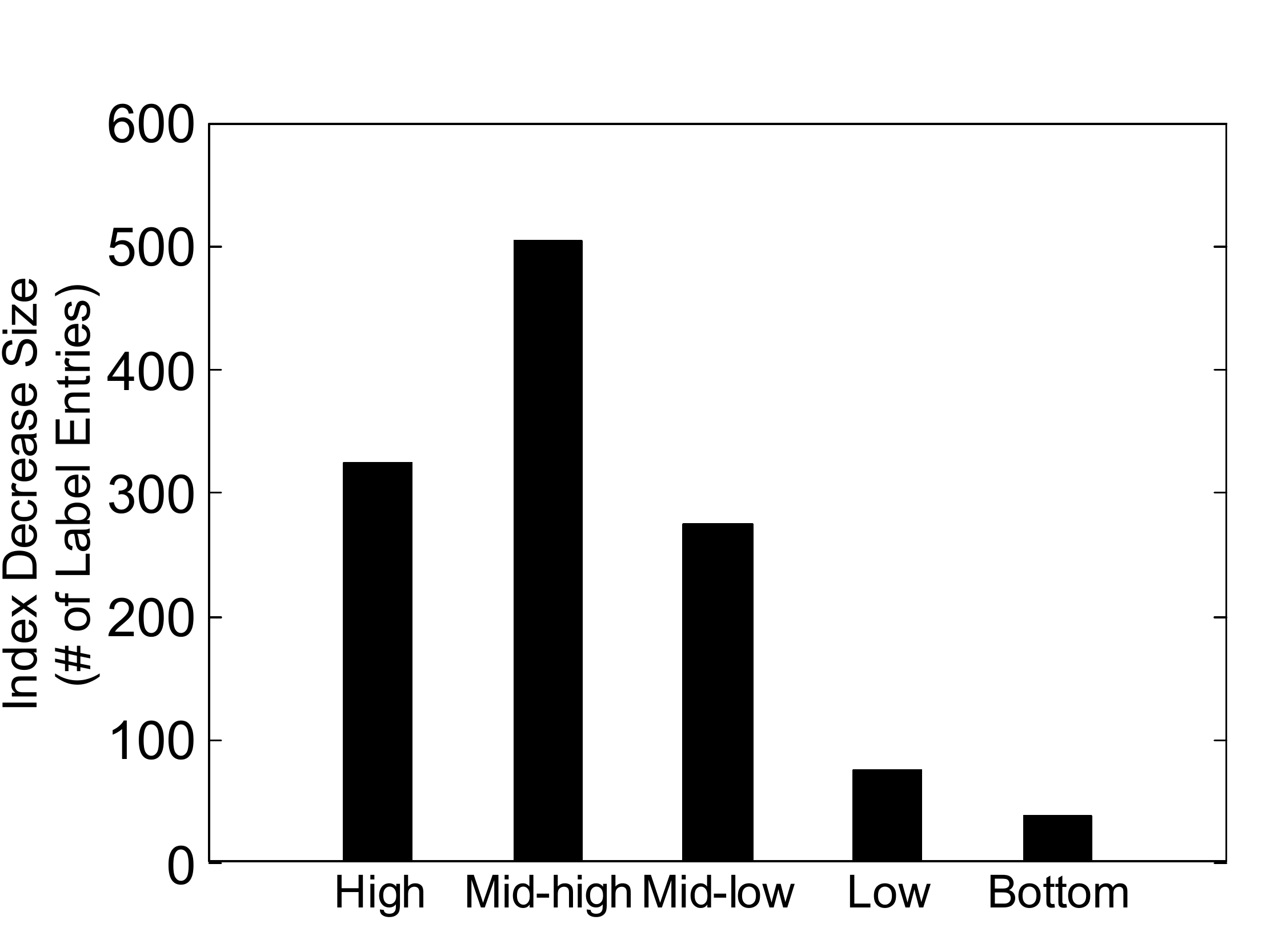}     
		}
	\end{center}
	\vspace{-5mm}
	\caption{Average Update Time (ms) and Index Decrease Size (Number of Label Entries) of Decremental Maintenance (Graph G04).} \label{fig:upd_dec} 
	\vspace{-3mm}
\end{figure}

\subsubsection{Index Size} \autoref{fig:upd_t_s}(b) provides the average increase size in label entry of incremental update. The number of newly added label entries for each update is from 443 to 1447 for the first seven graphs (except for WKT) under the minimality approach. As each label entry takes up 64 bits. The increased index size for these graphs is from 3.5KB to 11.5KB, and similarly, it is 135KB for graph WKT. Among most of the graphs, the increased index size accounts for less than 0.01\% of the original index. The results show that the difference in the increased label size between minimality and redundancy approaches is minor. Compared with the significant difference in their update time, the redundancy approach is considered ideal. For graphs WAR and WSR, their index increase sizes under redundancy only account for 5.7$\times 10^{-5}$ and 3.4$\times 10^{-6}$ of the size of their original indexes. Overall, the results illustrate that the label size under \textsc{Inc}SCT grows slowly.

From the data in \autoref{fig:upd_dec}(b), it is apparent that the deletion of high-degree edges may cause more label entries deleted. A large number of unaffected label entries are removed and recovered later.

\subsection{Case Study}
\label{sect:casestudy}
The real network dataset MAHINDAS \cite{nr} used for the case study is in the Economic Networks category. In this graph, vertices represent accounts, and edges represent transactions between them. The vertex size denotes the shortest cycle counting. The bigger a vertex, the more the shortest cycles pass through it. The vertex color denotes the shortest cycle length. The darker a vertex, the longer the shortest cycles. Figure~\ref{fig:casestudy} shows a subgraph centering at vertex 169, which all the shortest cycles through vertex 169 are listed. The figure matches the subgraph model in Figure~\ref{fig:application}. Vertices $281$, $241$, $169$, $1159$, and $888$ are filtered out to be suspicious accounts for anomaly detection, e.g., money laundering. Based on these candidates derived from our algorithm, we could further analyse whether there is an exact case for the financial crime by enumerating such cycles or paths~\cite{qiu18,peng19}

\begin{figure}[htb]
    \centering
    \vspace{-3mm}
    \includegraphics[scale=0.13]{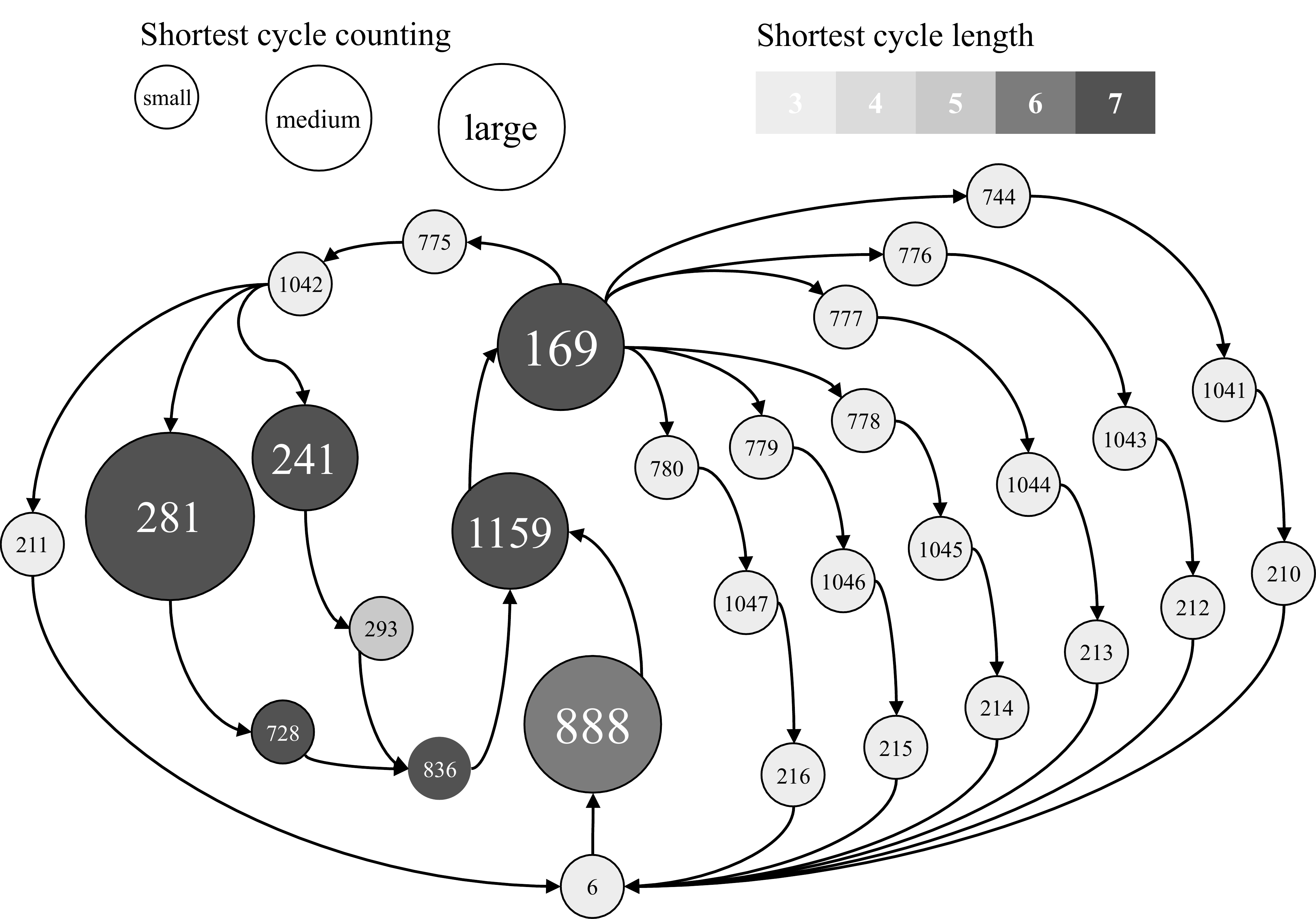}
    \vspace{-3mm}
    \caption{A subgraph centering at vertex $169$ with suspicious criminal nodes from a real economic network MAHINDAS.}
    \vspace{-3mm}
    \label{fig:casestudy}
\end{figure}
\section{Related Work}
\label{sec:related}

\subsubsection*{Counting Patterns}
\eat{Counting the number of specific patterns is a fundamental problem in graph theory.} Some counting problems are \#P-complete~\cite{valiant79}. For instance, counting the number of all the simple cycles through a given vertex or the number of simple paths between two vertices. 
The problem of counting the simple paths or cycles with a length constraint of $l$, parameterized by $l$, is considered to be \#W[1]-complete~\cite{flum04}. Many works proposed various algorithms to answer such problems. ~\cite{bezakova18} could answer the number of shortest paths in $O(\sqrt{n})$ where $n$ is the number of vertices, but the graph must be planar and acyclic. 
~\cite{giscard19} can also answer the number of shortest cycles, but the algorithm is too slow even when the graph is not so large. Many other similar works share the shortcomings like the graph type constraint or the undesirable query time.

\subsubsection*{Enumerating Patterns}
Enumerating the number of specific patterns is another method to produce the count. However, it is supposed to take much more time than direct counting \cite{giscard19}. For instance, when counting shortest cycles, one of the primary reasons is that the shortest cycle length is unknown in advance, and obtaining it requires the employment of a BFS-like method whose running time is already longer than the query time of hub labeling. Another reason is that enumeration requires finding all the vertices along with the cycles, which is unnecessary for the counting problem. With the help of a hot-point based index, \cite{qiu18} can output the newly-generated constraint cycles upon the arrival of new edges. \cite{peng19} is shown to have better time performance than \cite{qiu18} on enumerating constraint paths and cycles.

\subsubsection*{Dynamic Maintenance for 2-hop Labeling} To adapt to the dynamic update of the network, some works~\cite{akiba14,d2019fully,qin17} proposed dynamic algorithms to deal with edge insertion and deletion. For the edge insertion, a partial BFS for each affected hub is started from one of the inserted-edge endpoints and creates a label when finding the tentative distance is shorter than the query answer from the previous index. For the edge deletion, the state-of-the-art solution is to find the affected hubs followed by removing out-of-date labels, then recover labels. The time cost sometimes is acceptable compared with reconstructing the whole index. But it is much slower than that of incremental update. 

\section{Conclusion}
\label{sect:conclusion}
We investigate the problem of counting the number of shortest cycles through a vertex. A 2-hop labeling based algorithm is proposed for shortest cycle counting query. The comprehensive experiments demonstrate that our algorithms could achieve up to two orders of magnitude faster than state-of-the-art. We also present an update algorithm to maintain the index for edge updates. Our comprehensive evaluations verify the effectiveness and efficiency of the algorithms.

\section*{Acknowledgment}
Xuemin Lin is supported by ARC DP200101338. Wenjie Zhang is supported by ARC FT210100303 and DP200101116. Ying Zhang is supported by ARC DP210101393.

\newpage
\balance
\bibliographystyle{IEEEtran}
\bibliography{main}{}

\end{document}